\documentclass[11pt]{article}

\usepackage[margin=1in]{geometry}
\usepackage{amsmath,amssymb,amsthm,bm}
\usepackage{mathtools,bbm}
\usepackage{comment,booktabs,multirow}
\usepackage{enumitem}
\usepackage{graphicx,xcolor}
\usepackage{subfigure,float}
\usepackage{parskip}
\newtheorem{theorem}{Theorem}
\newtheorem{corollary}{Corollary}
\newtheorem{lemma}{Lemma}

\usepackage[size=tiny]{todonotes}

\newcommand{\Pp}{\mathbb P}
\newcommand{\E}[1]{\mathbb E\!\left[#1\right]}
\newcommand{\Var}[1]{\operatorname{Var}\!\left(#1\right)}
\newcommand{\ind}{\mathbf 1}
\newcommand{\Bernoulli}{\operatorname{Bernoulli}}
\newcommand{\Binomial}{\operatorname{Binomial}}

\allowdisplaybreaks

\usepackage[showdeletions]{color-edits}
\addauthor[Weina]{ww}{purple}
\usepackage[bookmarks,colorlinks=true,plainpages=false,
            citecolor=red,
            linkcolor=blue,
            anchorcolor=red,
            urlcolor=blue]{hyperref}

\begin{document}

\title{Approaching Optimal Scaling of Queue Size for an Input-Queued Switch}
\author{
Taha Ameen and Weina Wang
\thanks{
        The authors are with the Computer Science Department, Carnegie Mellon University, Pittsburgh, PA 15213, USA. E-mails: \{\texttt{tameen@cmu.edu, weinaw@cs.cmu.edu}\} 
    }
}
\date{}
\maketitle

\begin{abstract}
    We study the scaling of the expected total queue length in an $n\times n$ input-queued switch. 
    It has been conjectured that the optimal scaling is $O(n/(1-\rho))$, up to logarithmic factors, where $\rho$ is the system load.
    This conjecture has been established in the traditional heavy-traffic regime and in certain sub-regimes of the large-system regime.
    But it remains unresolved over a wide range of load scalings in the large-system regime.
    In this paper, we develop a family of scheduling policies parameterized by a positive integer $k\ge 2$.
    We show that under uniform traffic, our policy achieves an expected total queue length of $O\bigl({n\log\max\{n,1/(1-\rho)\}}\big/{(1-\rho)^{2^k/(2^k-1)}}\bigr)$, as long as $\rho\to 1$ in the large-system regime.
    As $k$ increases, the exponent $2^k/(2^k-1)$ approaches $1$, bringing our bound arbitrarily close to the conjectured scaling.
    The expected amortized scheduling cost of the policy is $O(n^3 k)$ per time slot.
    Our policy is of a batching type, with a novel recursive design that allows an arbitrary number of batching levels.
    This design principle may also be useful for other scheduling problems in which combinatorial constraints restrict which jobs can be served simultaneously.
\end{abstract}


\noindent\textbf{AI disclosure:} The ideas in this paper are the authors' own, including the $k$-level batching approach and the complete proof for the case $k=3$, which contains all the essential ingredients of the argument. In that proof, AI assistance was used only for routine steps. The extension to general $k$ (Appendix~\ref{sec:extension}) follows the same argument and is primarily technical. GPT 6 Astra was used to write the formal proof in Appendix~\ref{sec:extension}. For this, the authors provided a proof outline along with instructions for the steps should generalize. The authors have verified the extension and take full responsibility for the correctness of all results in the paper.

\setcounter{tocdepth}{2}
\tableofcontents

\section{Introduction}
\addtocontents{toc}{\protect\setcounter{tocdepth}{1}}

The input-queued switch is a canonical model of scheduling under matching constraints. It originated as an architecture for high-speed routers~\cite{mckeown1999throughput}, and has since become a central example of switched networks and stochastic processing networks, which also model wireless networks and other systems with interacting resources~\cite{harrison2000brownian, shah2014optimal}. 
An $n\times n$ input-queued switch has $n$ input ports and $n$ output ports, and each job must travel from a specific input to a specific output. In each time slot, every port can handle at most one job, so the jobs served in a slot form a matching between inputs and outputs. Serving one job thus rules out every other job that shares its input or its output, and the scheduler must coordinate service across all ports at once.
The capacity region of this system and throughput-optimal scheduling policies such as MaxWeight have long been known~\cite{tassiulas1992stability, mckeown1999throughput}.
However, the optimal scaling of the expected total queue length, which captures the delay performance, is much less understood.

A conjecture put forward in \cite{shah2011optimal} is that the optimal scaling of the expected total queue length is $O(n/(1-\rho))$, up to logarithmic factors, where $\rho$ is the system load.
This conjecture is expected to hold both in the traditional heavy-traffic regime, where $n$ is fixed and $\rho\to 1$, and in the large-system regime, where $n\to\infty$ and $\rho$ may scale with $n$.
In this paper, we sometimes omit the wording of up to logarithmic factors for brevity in informal discussions.

This conjectured scaling can be understood from several perspectives.
One interpretation is based on the number of resources in the system.
Although an $n\times n$ switch has $n^2$ queues, it has only $2n$ resources, corresponding to the input and output ports.
The conjecture suggests that, under a well-designed scheduling policy, the scaling of the expected total queue length is linear in the number of resources, rather than in the number of queues.
A similar phenomenon has been observed in bandwidth-sharing networks, where the scaling is linear in the number of links (resources), rather than the number of routes (queues) (see, e.g., \cite{KanKelLee_09,YeYao_16,WanMagSri_22}).
An important implication that becomes clearer from this interpretation is that the mean job delay is nearly independent of the switch size.
In particular, under uniform traffic, the conjecture implies that the mean job delay scales as $O(1/(1-\rho))$, up to logarithmic factors, by Little's law and the fact that the total arrival rate is $n\rho$.

Another interpretation that we find useful is from the perspective of the cost imposed by the combinatorial matching constraint.
To see this, imagine an idealized system in which all jobs arriving at input port $i$ are placed in a single queue, and each input port has the capacity of serving one job per time slot, independent of other input ports.
Under uniform traffic, each such queue then has an expected length of $O(1/(1-\rho))$, yielding an expected total queue length of $O(n/(1-\rho))$, which matches the conjectured scaling.
This is in line with the classical thinking that queueing can be viewed as a cost of the capacity constraints.
However, in the switch system, there is also a combinatorial matching constraint in addition to the capacity constraints.
This can prevent some input ports from serving a job in a time slot even when they have backlogged jobs, which is what makes the policy design challenging.
The conjecture suggests that there exists a policy design that overcomes this challenge, and achieves the same queue-length scaling as the idealized system with only capacity constraints.

This conjecture has motivated an extensive line of research over the years \cite{shah2014optimal,maguluri2016heavy,shah2016queue,xu2020improved,weng2022algorithm,jhunjhunwala2022lowcomplexity}.
In the traditional heavy-traffic regime, the conjectured scaling has been established first by \cite{shah2014optimal} using a policy that emulates an insensitive bandwidth sharing system, later by \cite{maguluri2016heavy} for the classical MaxWeight policy, and then also by \cite{jhunjhunwala2022lowcomplexity} for several low-complexity policies.

In the large-system regime, however, the conjecture has only been resolved at the two ends of the load range.
One end of the range is when the load $\rho$ is fixed.
If $\rho<1/2$, maximal scheduling already gives a scaling of $O(n)$~\cite{neely2008maximal}.
For any fixed $\rho<1$, the batching policy proposed in \cite{neely2007logarithmic} and further analyzed in \cite{shah2011optimal} achieves a scaling of $O(n\log n/(1-\rho)^2)$, which matches the conjectured scaling up to logarithmic factors.
At the other end, when $1/(1-\rho)\ge n^2$, the policy of \cite{shah2014optimal} achieves the conjectured scaling $O(n/(1-\rho))$.
Despite further progress in~\cite{shah2016queue,xu2020improved,weng2022algorithm}, the intermediate regime remains unresolved.
In particular, when $1/(1-\rho)\to\infty$ and $1/(1-\rho)=o(n^2)$, no existing policy is known to achieve the conjectured scaling.
This regime is arguably harder since although $\rho\to 1$, the heavy-traffic effect is less pronounced, making it harder to overcome the effect of the combinatorial matching constraint.
We summarize the best previously known scalings across different traffic regimes in the table below.
\begin{table}[ht]
\centering
\begin{tabular}{lcccc}
\toprule
Regime & $\frac{1}{1-\rho}\le n^{3/4}$ & $n^{3/4}\le \frac{1}{1-\rho}\le n^{7/4}$ & $n^{7/4}\le \frac{1}{1-\rho}\le n^{2}$ & $\frac{1}{1-\rho}\ge n^{2}$\\
\midrule
Scaling & $O\left(\frac{n\log\max\{n,1/(1-\rho)\}}{(1-\rho)^{4/3}}\right)$ & $O\left(\frac{n^{5/4}\log\max\{n,1/(1-\rho)\}}{1-\rho}\right)$ & $O(n^3)$ & $O\bigl(\frac{n}{1-\rho}\bigr)$\\[6pt]
Reference & \cite{xu2020improved} & \cite{weng2022algorithm} & \cite{shah2014optimal}  & \cite{shah2014optimal}\\  
\bottomrule
\end{tabular}
\end{table}

\subsection{Our result}
In this paper, we make progress towards resolving the conjecture in the large-system regime.
We develop a family of scheduling policies, parameterized by a positive integer $k$, that achieves the following scaling of the expected total queue length under uniform traffic whenever $1/(1-\rho)\to\infty$~\footnote{
Our result can be used to obtain a scaling bound for a fixed $\rho<1$.
Specifically, we can add virtual jobs to inflate the load, say, to $\rho'$ such that $1/(1-\rho')=\log\log n$.
Then applying our result to this system gives the scaling $O(n(\log n)(\log\log n)^2)$.}:
\[
O\left(\frac{n\log\max\{n,1/(1-\rho)\}}{(1-\rho)^{2^k/(2^k-1)}}\right).
\]
In this bound, the exponent of the denominator $(1-\rho)$ is $4/3$ when $k=2$, $8/7$ when $k=3$, and approaches $1$ as $k\to\infty$.
Therefore, for any fixed $\delta>0$, we can choose a sufficiently large $k$ to obtain a scaling of
\[
O\left(
\frac{n\log\max\{n,1/(1-\rho)\}}
{(1-\rho)^{1+\delta}}
\right).
\]
Thus, our result comes arbitrarily close to the conjectured scaling.

\subsection{Policy design}

Our policy design builds on the batching policy \cite{neely2007logarithmic,shah2011optimal}.
The main idea of the batching policy is to collect job arrivals over a window of $T$ time slots and serve them only after the window ends.
Once the window ends, the policy decomposes the collected jobs into as few matchings as possible to serve them.
This way, the policy avoids making scheduling decisions shortsightedly.
Instead, it waits until enough jobs have accumulated so it can form matchings that serve many jobs simultaneously, thus making better use of the resources.
To handle stochastic fluctuations and maintain stability, the window length $T$ needs to be at least on the order of $\Theta(\log n/(1-\rho)^2)$.
This means that most jobs must wait for at least $\Theta(\log n/(1-\rho)^2)$ time slots, rather than the conjectured $O(1/(1-\rho))$.
Since this window length becomes the bottleneck of the performance, subsequent papers~\cite{shah2016queue,xu2020improved,weng2022algorithm} have mainly refined the batching policy by starting service earlier within the window, achieving improved queue-length scalings.

\paragraph{Our policy.}
Our policy is also of a batching type.
Compared to existing batching-type policies, the main novelty lies in our \emph{recursive} design that allows an arbitrary number of batching levels.
A formal description of the policy is provided in Section~\ref{sec:policy-formal}.
Below we describe the main high-level ideas of the policy design.

Roughly speaking, our level-$1$ batching is similar to the original batching policy, but we shorten the batching window $T$ and we serve only perfect matchings.
With this shorter window, level-$1$ batching alone cannot serve all job arrivals at a sufficient rate to maintain stability.
We then introduce level-$2$ batching, which groups multiple level-$1$ windows into a larger block.
We collect the jobs left unserved at level $1$ and serve them at time slots dedicated to level-$2$ serving.
By pooling these jobs, the policy can form matchings more efficiently and serve the jobs with better resource utilization.
We can recursively extend this approach to $k$ levels of batching for any $k\ge 2$.
Jobs served at higher levels experience longer delays, but the fraction of such jobs becomes progressively smaller.
Overall, our design guarantees that the larger $k$ is, the better queue-length scaling we achieve.

Regarding the computational complexity, the $k$-level policy can be implemented with $O(n^3 k)$ expected amortized scheduling cost, per time slot.

\paragraph{Technical challenge.}
A key technical challenge in our analysis is to guarantee that, at every level except the last, a sufficient number of perfect matchings can be formed to serve the jobs.
At level~$1$, this is guaranteed by the randomness of job arrivals.
At a higher level, however, the jobs consist of those left unserved at the preceding level, and their distribution may not be sufficiently balanced to form enough perfect matchings.
To address this difficulty, we introduce a mechanism in our policy that randomly assigns a small fraction of arriving jobs directly to levels $2,3,\dots,k-1$.
We prove that this additional randomness, together with certain structural properties of the residual jobs, guarantees the existence of the required number of perfect matchings at each level.

\section{Problem setup}\label{sec:model}
\addtocontents{toc}{\protect\setcounter{tocdepth}{2}}

Consider an $n\times n$ input-queued switch that operates in discrete time $t\in \{0,1,2,\cdots\}$. Each input--output pair $(i,j)$, where $1 \leq i,j \leq n$, has its own queue. Let $Q_{ij}(t)$ be the length of this queue at the beginning of slot $t$, and let $A_{ij}(t)$ be the number of jobs that arrive to it during that slot. We assume
\[
    A_{ij}(t)\sim\Bernoulli(\rho/n),\qquad 0<\rho<1,
\]
independently across queues and slots. Each job requires one slot of service and leaves immediately after service. In each slot, service is subject to the switch constraint, i.e., each input port can transmit at most one job, and each output port can receive at most one job.

A schedule $\bm\sigma(t)\in\{0,1\}^{n\times n}$ specifies the queues selected for service in slot $t$, with $\sigma_{ij}(t)=1$ if queue $(i,j)$ is selected. The switch constraint above means that a feasible schedule is a matching between inputs and outputs, i.e., it satisfies:
\begin{align}\label{eq:matching-constraints}
    \sum_{j=1}^n\sigma_{ij}(t) &\leq 1 \, ,
    \quad i = 1, 2, \cdots, n \, ,
    \\
    \sum_{i=1}^n\sigma_{ij}(t) &\leq 1 \, ,
    \quad j = 1, 2, \cdots, n \, . 
\end{align}
We use the convention that within each slot, the scheduler first observes the queue lengths, then new jobs arrive, and finally the selected jobs are served. Hence, the queue evolves as
\begin{align}\label{eq:queue-dynamics}
    Q_{ij}(t+1)=
    \big(
        Q_{ij}(t)+A_{ij}(t)-\sigma_{ij}(t)
    \big)^+,
\end{align}
where $x^+=\max\{x,0\}$. We start with $Q_{ij}(0)=0$ for all $i,j$ and denote the total queue size of the switch at the beginning of time slot $t$ as
\[
    Q_\Sigma(t)=\sum_{i,j=1}^n Q_{ij}(t).
\]
We are interested in bounding $\mathbb{E} \big[ Q_\Sigma(t) \big]$.

\paragraph{Scaling regime.}
We consider the scaling regime where $n \to \infty$ and $\rho \to 1$ jointly.
We do not put any constraints on how fast $\rho\to 1$ as $n\to \infty$.
For convenience, let $\varepsilon=1-\rho$ denote the excess capacity. 
Then we have $\varepsilon\to 0$ as $n\to\infty$, i.e., $\varepsilon=o(1)$.

\paragraph{Multi-graph representation.}
It is convenient to view a set of jobs as a bipartite multi-graph. Its two sides are the sets $\mathcal I$ and $\mathcal O$ of input and output ports, each of size $n$, and each job in queue $(i,j)$ is an edge between $i\in\mathcal I$ and $j\in\mathcal O$. Distinct jobs in the same queue are represented as parallel edges. In this view, a feasible schedule corresponds to a matching of the graph. 

\paragraph{General notation.} 
Standard Bachmann–Landau notation ($O(\cdot), o(\cdot), \Theta(\cdot),\dots$) is used throughout, and it is implicit that $n \to \infty$. 
For an event $E$, we use $\ind_E$ for the indicator of $E$. For a graph $G$, we use $\deg_G(v)$ to denote the degree of a vertex $v$, and we use $E(G)$ to denote the set of edges of the graph. We also use $e_G(X,Y)$ to denote the number of edges between $X\subseteq\mathcal I$ and $Y\subseteq\mathcal O$, counted with multiplicity. An \emph{$s$-factor} of $G$ is a spanning subgraph in which every port has degree equal to $s$.

\section{Main result and scheduling policy}\label{sec:result}

Our main result concerns the \emph{$k$-level batching policy}, which we describe in Section~\ref{sec:policy-formal}. The policy groups time slots into nested blocks over $k$ levels. At each of the first $k-1$ levels, it serves an $s$-factor of each batch using $s$ perfect matchings, where $s$ depends on the level, and passes the remaining jobs to the next level. At the last level, it splits the jobs that reach it into matchings and serves them in first-come, first-served order. The following theorem captures its performance.

\begin{theorem}\label{thm:main}
    Fix an integer $k \geq 2$. Starting from an empty switch, the $k$-level batching policy satisfies
    \[
        \sup_{t\geq 0} \, 
        \mathbb{E} \big[ Q_\Sigma(t) \big]
        = O\left(
            \frac{n\log(n/\epsilon)}{\varepsilon^{2^k/(2^k-1)}}
        \right)
    \]
    for all sufficiently large $n$, where the constant in $O(\cdot)$ depends only on $k$.
\end{theorem}

The theorem places no restriction on how fast $\varepsilon\to0$ as $n\to\infty$. In Section~\ref{sec:proof-three}, we prove the theorem for the case $k=3$, which captures the core ideas. We then extend the argument to every fixed $k$ in Appendix~\ref{sec:extension}.

The exponent $2^k/(2^k-1)$ in Theorem~\ref{thm:main} decreases to one as $k$ grows. When $\varepsilon$ decays at least polynomially in $n$, the logarithmic factor can also be absorbed into a small power of $1/\varepsilon$. This gives the following bound.

\begin{corollary}\label{cor:delta}
    Fix $\delta>0$ and $c>0$, and suppose that $\varepsilon \leq n^{-c}$. Let $k\ge2$ be the smallest integer satisfying $1/(2^k-1)<\delta$. Starting from an empty switch, the $k$-level batching policy satisfies
    \[
        \sup_{t \geq 0} \,
        \mathbb{E} \big[ Q_\Sigma(t) \big]
        = O\left(
            \frac{n}{\varepsilon^{1+\delta}}
        \right)
    \]
    for all sufficiently large $n$, where the constant in $O(\cdot)$ depends only on $\delta$ and $c$.
\end{corollary}

\begin{proof}
    Since $\varepsilon\le n^{-c}$, we have that $\log(n) \leq c^{-1} \log(1/\varepsilon)$. Therefore, 
    $\log(n/\varepsilon)\le (1+c^{-1})\log(1/\varepsilon)$.
    Theorem~\ref{thm:main} therefore gives
    \[
        \sup_{t\ge0}\mathbb{E}\big[Q_\Sigma(t)\big]
        = O\!\left(
            \frac{n\log(1/\varepsilon)}
                 {\varepsilon^{1+1/(2^k-1)}}
        \right)
        = O\!\left(
            \frac{n}{\varepsilon^{1+\delta}}
        \right).
    \]
    The last step uses that $\log(1/\varepsilon) =O(\varepsilon^{-a})$ for every fixed $a>0$, together with $1/(2^k-1)<\delta$.
\end{proof}

Thus, when $\varepsilon\leq n^{-c}$, a suitable choice of $k$ ensures that the expected queue size is within a factor $O(\varepsilon^{-\delta})$ of the lower bound $n/\varepsilon$, for any fixed $\delta>0$. In this sense, the $k$-level batching policy comes arbitrarily close to the conjectured $n/\varepsilon$ scaling.

\subsection{The~\texorpdfstring{$k$}{}-level batching policy}\label{sec:policy-formal}

We now describe our $k$-level batching policy for a general $k\geq 2$. The policy groups time slots into blocks at different time scales, and assigns a \textit{level} to each job, time slot, and block.   

\paragraph{Levels of blocks.}
We group consecutive time slots into \emph{level-1 blocks}, with each block containing $L_1$ slots. Next, we group $L_2$ consecutive level-1 blocks into a single \emph{level-2 block}, and so on. Thus, for $j \in \{2,\cdots, k\}$, every $L_j$ consecutive level-$(j\!-\!1)$ blocks form a \emph{level-$j$ block}, which contains $T_j := L_1L_2\cdots L_j$ slots.  The lower panel of Figure~\ref{fig:policy} illustrates this for $k=3$. 

\paragraph{Levels of service slots.} For $j \in \{1,2,\cdots,k\}$, each level-$j$ block reserves $S_j$ of its slots as \emph{level-$j$ service slots}. The positions of these level-$j$ service slots in the block are determined in increasing order, i.e., level 1 takes the first $S_1$ slots of each level-1 block, level 2 takes the first $S_2$ remaining slots of each level-2 block, and so on. The $S_j$'s are chosen so that all time slots are assigned a level. The lower panel of Figure~\ref{fig:policy} illustrates the level-$j$ slots and their position, with slots marked and shaded according to their level. 

\begin{figure}[t]
    \centering
    \includegraphics[width=0.99\linewidth]{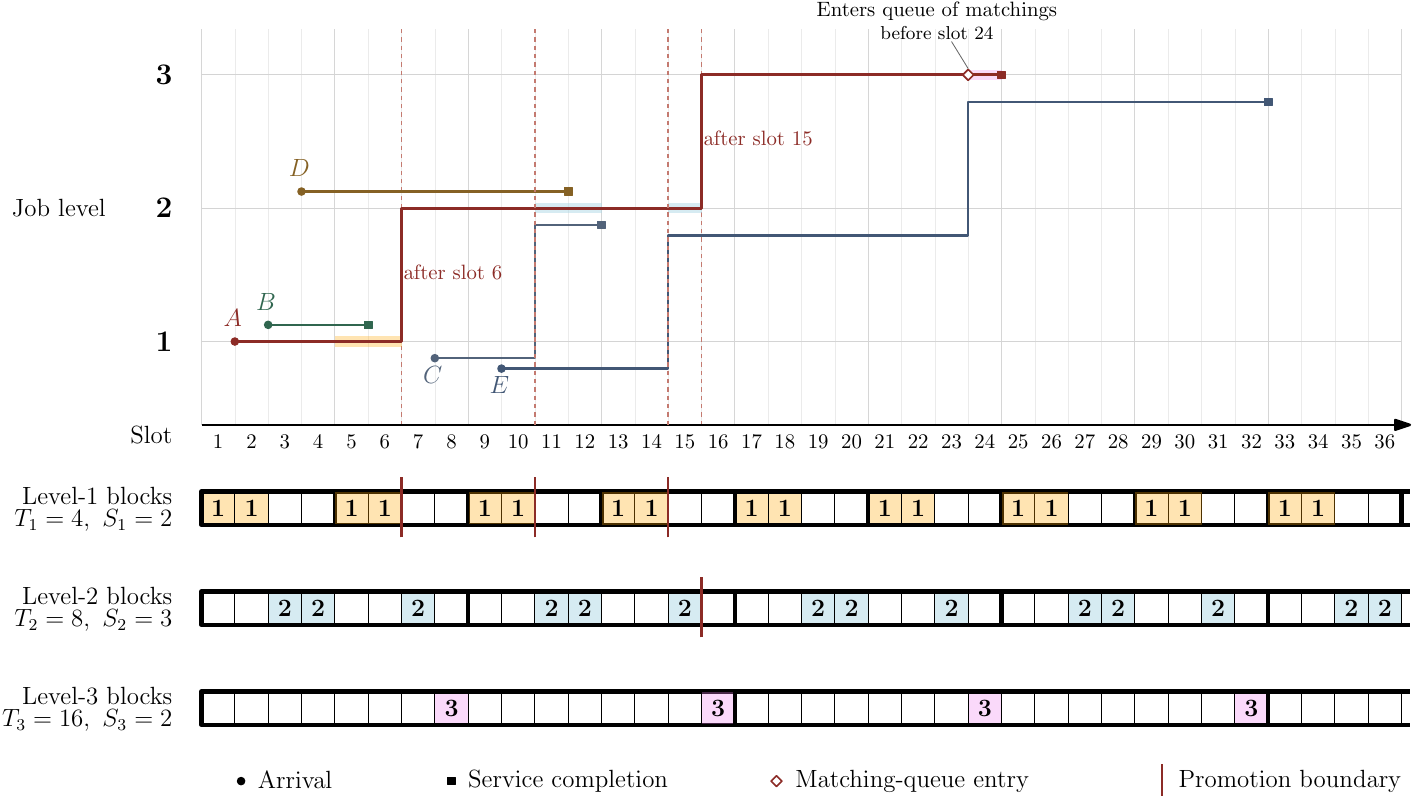}
    \caption{Illustration of the three-level batching policy with $L_1=4$, $L_2=L_3=2$, $S_1=2$, $S_2=3$, and $S_3=2$. The lower panel shows time slots grouped into blocks at different levels, and the level of each service slot is labeled and shaded. The upper panel illustrates five job trajectories: dots represent arrivals, jumps represent promotions, and squares represent service completion. At levels $j=1,2$, a job is eligible in the level-$j$ block immediately after its original arrival block at that level. If unserved, it is promoted immediately after the last level-$j$ service slot of its eligible block. The slots where promotions occur are marked as well.}
    \label{fig:policy}
\end{figure}
\paragraph{Levels of jobs.}
At any time, each job in the switch has a \emph{level} in $\{1,\cdots,k\}$. Jobs arrive with a randomly chosen level $j \in \{1,2,\cdots,k-1\}$, and may get promoted later to higher levels. At any time, level-$j$ service slots serve only jobs that are currently \emph{level-$j$ jobs}. Most jobs start at level~1, however, a small random fraction of arrivals is set aside for each intermediate level, to keep the batches there suitably spread. Specifically, the initial level $M$ of an arriving job equals $r+1$ with probability $\theta_r$ (defined in~\eqref{eq:def-theta-r}), for $r=1,\cdots,k-2$, and equals 1 otherwise, independently of the arrivals and of other jobs. After arrival, a job's level changes only through \emph{promotion}: a job that is not served at level $j\leq k-1$ is promoted to level $j+1$, at a time specified in the service schedule below. Thus, a level-$j$ job either started at level $j$ or was promoted from level $j-1$. No job starts at level $k$, and a job that is promoted to level $k$ stays there until it is served. The upper panel of Figure~\ref{fig:policy} illustrates an example trajectory for 5 jobs, including their time slot of arrival, departure and sequence of promotions.

\paragraph{Eligibility for service.}
For each $j \in \{1,\cdots, k - 1\}$, a level-$j$ job becomes \emph{eligible} for service in the level-$j$ service slots of the level-$j$ block immediately after the one in which it arrived. For example, a level-$j$ job that arrives early in a level-$j$ block must wait for the next level-$j$ block to receive service, even if it arrived before the current block's first level-$j$ service slot. 

This eligibility rule applies to promotions as well. In particular, a job whose level is promoted from $j$ to $j+1$ (with $j+1 \leq k-1$), becomes eligible for service only in the level-$(j+1)$ service slots of the level-$(j+1)$ block that is immediately after the level-$(j+1)$ block of the original arrival of the job. The policy ensures that the promotion to level-$(j+1)$ happens before any of these level-$(j+1)$ service slots have passed (the specifics are described in the service schedule below). A level-$k$ job remains eligible until it is served.
Note that in Figure~\ref{fig:policy}, jobs are only served in a time slot where they are eligible. Further, the trajectory of a job is shaded for the time slots in which a job is eligible for service but does not receive service.

\paragraph{Service schedule.} $ \ $
\begin{itemize}[leftmargin=2em]
    \item \textbf{At levels $1$ to $k-1$.
     } 
     Fix $j\leq k-1$ and a level-$j$ block. Just before the first level-$j$ service slot of the block, the policy collects the level-$j$ jobs that are eligible in this block, i.e., those that arrived during the previous level-$j$ block, into a \emph{batch}. If the graph of the batch contains an $S_j$-factor, the policy decomposes it into $S_j$ perfect matchings and serves one in each level-$j$ service slot of the block. Otherwise, these slots stay idle. Immediately after the last level-$j$ service slot, every unserved job of the batch is promoted to level $j+1$. Since promotion does not change a job's arrival slot, the promoted job is eligible for level-$(j+1)$ service in the level-$(j+1)$ block immediately after the one in which it arrived. 
    
    \item \textbf{At level $k$.\ } 
    Level $k$ cannot pass leftovers on, so it must serve every job that reaches it, whether or not they form perfect matchings. Just before the first level-$k$ service slot of each level-$k$ block, the policy splits the batch of level-$k$ jobs whose original arrivals occurred during the previous level-$k$ block into as few matchings as possible. By K\H{o}nig's edge-coloring theorem, their number equals the maximum degree of the batch. These matchings join the back of a \emph{queue of matchings}. Each level-$k$ service slot serves the matching at the front of this queue, or stays idle if the queue is empty, so unserved matchings wait for later level-$k$ blocks.
\end{itemize}

\paragraph{Remarks.} In the description above, we have not specified which $S_j$-factor to serve, how to split a batch into matchings, and in which order to serve them. These are decided by fixed rules determined before observing the switch. Hence, at each level, the batches of different blocks are independent and identically distributed. Since the switch starts empty, the level-$j$ service slots of the first level-$j$ block stay idle, for each $j$.

\paragraph{Runtime.} 
The $k$-level policy can be implemented with $O(n^3 k)$ expected amortized scheduling cost, per time slot. To see this, note that selecting an $s$-factor requires a maximum-flow computation costing $O(n^3)$, and extracting a matching from the $s$-factor also costs $O(n^3)$ time. Finally, each level performs at most one $s$-factor computation per block, giving an amortized cost of $O(n^3)$ per level per time slot.

\paragraph{Choice of parameters.}
Let $C=108$ and $K=144$, and define
\(
    h :=\log(Kn/\varepsilon).
\)
The numbers of service slots are defined recursively, together with auxiliary quantities $a_j$, by
\begin{align}
     S_1 &= \big\lfloor\rho T_1-C\sqrt{T_1h}\big\rfloor \, ,
    &a_1 &= \rho T_1-S_1 \, , \label{eq:service-s1}
    \\
     S_j &= \big\lfloor L_ja_{j-1}-C\sqrt{T_jh} \big\rfloor \, ,
    &a_j &=L_ja_{j-1}-S_j
     \ \ \text{for } j = 2,\dots, k-1 \, .\label{eq:service-sj}
\end{align}
The intuition behind this definition is as follows: each port receives $\rho T_1$ jobs per level-1 block on average, and $S_1$ is set below this mean by a margin of order $\sqrt{T_1h}$ to absorb random fluctuations. An $S_1$-factor then exists with high probability, and about $a_1$ jobs per port are passed up to the next level. 

In the same way, level $j$ pools about $L_ja_{j-1}$ jobs per port, and $S_j$ is selected to be a margin of order $\sqrt{T_jh}$ below this, so that about $a_j$ jobs per port are passed up. Since $C\sqrt{T_jh}\leq a_j<C\sqrt{T_jh}+1$, this is a small fraction of the $\rho T_j$ arrivals per port in a level-$j$ block. These numbers fit within the blocks: by induction on $j$, each level-$j$ block with $j\leq k-1$ leaves
\begin{equation}\label{eq:remaining-slots}
    T_j-\sum_{i=1}^{j}\frac{T_j}{T_i}S_i=\varepsilon T_j+a_j
\end{equation}
slots for higher levels, so each level-$k$ block contains $S_k$ level-$k$ service slots, where
\begin{equation}\label{eq:service-bk}
    S_k\coloneqq L_k\big(\varepsilon T_{k-1}+a_{k-1}\big)=\varepsilon T_k+L_ka_{k-1} \, .
\end{equation}
Each job starts off in initial level $r+1$ with probability $\theta_r$, where 
\begin{equation} \label{eq:def-theta-r}
    \theta_r := a_r/(2\rho T_r).
\end{equation}
This holds for $r=1,\cdots,k-2$. With remaining probability $1 - \sum_{r=1}^{k-2}\theta_r$, the job starts off in initial level $1$.
Finally, we choose the block lengths as
\[
    L_1 = \big\lceil Kh\varepsilon^{-\alpha_k} \big\rceil \, ,
    \qquad
    L_j=\big\lceil K\varepsilon^{-\alpha_k/2^{j-1}}\big\rceil
    \ \ \text{for } j = 2,\dots, k \, ,
\]
where $\alpha_k=2^k/(2^k-1)$ is the exponent in Theorem~\ref{thm:main}. These lengths are chosen to balance the levels, so that each level contributes approximately $nh\varepsilon^{-\alpha_k}$, to the expected queue size, and the top-level blocks have length $T_k$ of order $h/\varepsilon^2$.

\section{Proof of Theorem~\texorpdfstring{\ref{thm:main}}{} when~\texorpdfstring{$k=3$}{}}\label{sec:proof-three}

Throughout the proof, the parameters are obtained from Section~\ref{sec:policy-formal} with $k=3$. In particular, the block lengths $L_j$ and the number of slots $S_j$ for level-$j$ service in a level-$j$ block are
\begin{gather}
    L_1=\lceil Kh\varepsilon^{-8/7}\rceil,\quad
    L_2=\lceil K\varepsilon^{-4/7}\rceil,\quad
    L_3=\lceil K\varepsilon^{-2/7}\rceil,\label{eq:three-blocks}
    \\
    S_1=\rho T_1-a_1,\quad S_2=L_2a_1-a_2,\quad
    S_3=\varepsilon T_3+L_3a_2.\label{eq:three-service}
\end{gather}
Thus, the number of time slots in each level satisfies the following bounds, which we use frequently throughout the proof.
\begin{equation}\label{eq:three-block-scales}
    \begin{aligned}
         Kh\varepsilon^{-8/7}    & \leq T_1 \leq 2Kh\varepsilon^{-8/7},\\
         K^2h\varepsilon^{-12/7} & \leq T_2 \leq 4K^2h\varepsilon^{-12/7},\\
         K^3h\varepsilon^{-2}    & \leq T_3 \leq 8K^3h\varepsilon^{-2}.
    \end{aligned}
\end{equation}
Lastly, note that a job starts at level 2 with probability $\theta=a_1/(2\rho T_1)$, and otherwise starts at level 1. Since the choice of parameters satisfies $a_1/(2\rho T_1) \to 0$, the quantity $\theta$ is a well defined probability for all sufficiently large $n$.

To prove the result for $k = 3$, we will bound the number of level-1 jobs, level-2 jobs and level-3 jobs present in the switch, at any time slot $t$. Specifically, let $X_j(t)$ denote the number of unserved level-$j$ jobs at the beginning of slot $t$. It follows that
\begin{align} \label{eq:stage-decomposition}
    Q_\Sigma(t) = X_1(t) + X_2(t) + X_3(t) \,,
\end{align}
since each job is always at one of the three levels. We will show that each term is $O(nh\varepsilon^{-8/7})$ in expectation, at any time $t$.

Table~\ref{tab:three-level-notation} in Appendix~\ref{apx:table} summarizes the main notation used in the three-level analysis and its supporting proofs.

\subsection{Analysis of level-1 jobs}\label{sec:level-one}

First, we give a direct bound for the number of level-1 jobs in the switch at any time $t$.

\begin{lemma}\label{lem:queue-length-one}
     The expected number of level-1 jobs in the switch satisfies
        \(
         \sup_{t\ge0}\E{X_1(t)} \leq 2n\rho T_1.
        \)
\end{lemma}
\begin{proof}
     Every job starting at level 1 is either served or promoted by the last level-1 service slot of the next level-1 block. Thus, at any time $t$, level-1 jobs come from at most the current and preceding arrival blocks. Counting all arrivals in these two blocks gives $\E{X_1(t)}\leq 2n\rho T_1$.
\end{proof}

Next, we will bound the number of jobs from a level-1 arrival block that enter level 2. To that end, fix a level-1 arrival block, and recall that each arriving job corresponds to an edge in the bipartite multi-graph.
Consider the next level-1 block, and let $G_1$ denote the graph of level-1 jobs eligible for service at the beginning of this block. By definition, this only includes the level-1 jobs that arrived in the previous level-1 block. For each input--output pair, a job arrives and starts at level 1 in each slot with probability $\rho(1-\theta)/n$, independently across pairs and slots, so the edge multiplicities of $G_1$ are independent $\Binomial(T_1,\rho(1-\theta)/n)$ random variables. Let $\mathcal F^{(1)}$ be the failure event that $G_1$ has no $S_1$-factor, and define
\[
    F_1 := S_1 \cdot \ind_{\mathcal F^{(1)}},
\]
i.e., $F_1$ counts the level-1 service slots that stay idle in this block. The following lemma bounds the probability of failure and the number of idle service slots.

\begin{lemma}[Level-1 failures and idle service slots]\label{lem:success-one}
    For every level-1 arrival block,
    \begin{enumerate}[itemsep=0.5em]
        \item[$\mathrm{(i)}$]
        \(
            \Pp(\mathcal F^{(1)})  \leq e^{-40h}
        \)
        \item[$\mathrm{(ii)}$]
         \(
            \E{F_1}\le T_1e^{-40h} \leq a_1.
        \)
    \end{enumerate}
\end{lemma}

Lastly, we relate the idle service slots to the number of jobs that enter level 2. Let $R_1$ be the graph of jobs from the previous level-1 block that remain in the switch, measured immediately after the last level-1 service slot of the current block. These are (i) jobs that were eligible for level-1 service in the current block, but did not receive service and were promoted to level 2, and (ii) jobs that arrived in the previous level-1 block with initial level 2. We denote
\[
    J_1 := |E(R_1)|,
    \qquad
    D_1^{(v)} := \deg_{R_1}(v),
\]
i.e., $J_1$ is the total number of jobs from the previous level-1 block that enter level 2, and $D_1^{(v)}$ counts these jobs at port $v$.

We can write $D_1^{(v)}$ explicitly as follows. Let $N^{(v)}(T_1)$ denote the total number of job arrivals at port $v$ during the previous level-1 block. Then
\begin{align}\label{eq:leftovers-one}
    D_1^{(v)} = N^{(v)}(T_1)-\rho T_1+a_1+F_1.
\end{align}
This follows because level-1 service removes $S_1-F_1=\rho T_1-a_1-F_1$ jobs at every port, and no job from this arrival block has yet received level-2 service at the instant defining $R_1$.

\begin{lemma}\label{lem:population-one}
    Starting from an empty switch, for any level-1 arrival block,
    \begin{enumerate}[itemsep=0.5em]
        \item[$\mathrm{(i)}$] For every port $v$, $\mathbb{E}\big[D_1^{(v)}\big]\leq 2a_1$.
        \item[$\mathrm{(ii)}$]
        \(
            \E{J_1} \leq 2na_1.
        \)
    \end{enumerate}
\end{lemma}
The proofs of Lemmas~\ref{lem:success-one} and~\ref{lem:population-one} are in Section~\ref{sec:supporting-proofs}.

\subsection{Analysis of level-2 jobs}\label{sec:level-two}
First, we use the bound on $J_1$ to bound the number of level-2 jobs in the switch at any time $t$.

\begin{lemma}\label{lem:queue-length-two}
The expected number of level-2 jobs in the switch satisfies
\(
 \sup_{t\ge0}\E{X_2(t)}\le4nL_2a_1.
\)
\end{lemma}
\begin{proof}
Fix a slot $t$ and consider the level-2 jobs present in the switch. These jobs leave level 2 no later than the last level-2 service slot of the next level-2 block after their arrival. They therefore arrived during the current or preceding level-2 block, together containing at most $2L_2$ level-1 blocks. For each such level-1 block, they are among the jobs counted by $J_1$: a job that starts at level 2 skips level 1, and a job promoted to level 2 has already finished its level-1 service. By Lemma~\ref{lem:population-one}(ii), their expected number is at most
\(
 (2L_2)\cdot(2na_1)=4nL_2a_1.
\)
If the preceding block does not exist, there are fewer blocks to count. This completes the proof.
\end{proof}

Next, we will bound the number of jobs promoted to level 3. Fix a level-2 arrival block. Consider the next level-2 block, and let $G_2$ be the graph of level-2 jobs eligible just before its first level-2 service slot. By the policy, $G_2$ is the union of the graphs $(R_{1,c})_{c=1,\ldots,L_2}$ from the previous block's constituent level-1 arrival blocks \footnote{Here, $R_{1,c}$ is the graph of jobs from the $c$-th constituent level-1 arrival block that either started at level 2 or were promoted from level 1 (as defined in Section~\ref{sec:level-one}).}. Let $\mathcal F^{(2)}$ be the event that $G_2$ has no $S_2$-factor, and define
\[
    F_2=\sum_{c=1}^{L_2}F_{1,c}+S_2\ind_{\mathcal F^{(2)}}.
\]
Here, $F_{1,c}$ counts the idle level-1 service slots for jobs from the $c$-th constituent level-1 arrival block. Thus $F_2$ counts the level-1 and level-2 service slots for these jobs that stay idle because their factor tests fail. Level 2 can serve jobs left unserved at level 1, but it still has only $S_2$ service slots. For this reason, we account for missed service at both levels when counting jobs promoted to level 3.

The following lemma bounds the probability of a level-2 failure and the expected number of these idle service slots. We remark that the randomized initial levels are important here: jobs that start at level 2 provide edges spread across the input--output pairs, which, together with degree balance, ensure an $S_2$-factor with high probability.

\begin{lemma}[Level-2 failures and idle service slots]\label{lem:success-two}
    For every level-2 arrival block,
    \begin{enumerate}[itemsep=0.5em]
        \item[$\mathrm{(i)}$]
        \(
            \Pp(\mathcal F^{(2)})\le(L_2+1)e^{-40h}
        \)
        \item[$\mathrm{(ii)}$]
        \(
            \E{F_2}\leq T_2(L_2+1)e^{-40h}
            \leq a_2.
        \)
    \end{enumerate}
\end{lemma}

Lastly, we relate the idle service slots to the number of jobs promoted to level 3. Let $R_2$ be the graph of jobs from the previous level-2 block that remain in the switch, measured immediately after the last level-2 service slot of the current block. These are precisely the jobs promoted to level 3. We denote
\[
     J_2=|E(R_2)|,\qquad D_2^{(v)}=\deg_{R_2}(v).
\]
Thus $J_2$ counts the jobs from the previous level-2 block that enter level 3, and $D_2^{(v)}$ counts these jobs at port $v$.

We can write $D_2^{(v)}$ explicitly as follows. Let $N^{(v)}(T_2)$ denote the total number of job arrivals at port $v$ during the previous level-2 block. Then
\begin{equation}\label{eq:leftovers-two}
    D_2^{(v)}=N^{(v)}(T_2)-\rho T_2+a_2+F_2.
\end{equation}
Indeed, by~\eqref{eq:three-service}, service at levels 1 and 2 removes $L_2S_1+S_2-F_2=\rho T_2-a_2-F_2$ jobs at every port, and no job from this arrival block has yet received level-3 service at the instant defining $R_2$.

\begin{lemma}\label{lem:population-two}
    Starting from an empty switch, for any level-2 arrival block,
    \begin{enumerate}[itemsep=0.5em]
        \item[$\mathrm{(i)}$] For every port $v$, $\mathbb{E}\big[D_2^{(v)}\big]\leq 2a_2$.
        \item[$\mathrm{(ii)}$] $\E{J_2}\le2na_2$.
    \end{enumerate}
\end{lemma}

The proofs of Lemmas~\ref{lem:success-two} and~\ref{lem:population-two} are in Section~\ref{sec:support-two}.

\subsection{Analysis of level-3 jobs}\label{sec:matching-queue}
We now bound the number of level-3 jobs in the switch. Recall that upon promotion to level 3, a job waits to be batched with all other jobs from its original level-3 arrival block that get promoted to level 3. It waits until just before the first level-3 service slot of the next level-3 block, where its batch is partitioned into matchings. These matchings join the back of the queue of matchings.

To study the evolution of the queue of matchings, we index level-3 arrival blocks by $b=0,1,\ldots$, and let $A_3(b)$ be the arrival process, i.e., the number of matchings formed from the jobs of block $b$ that reach level 3. Since the same rules are used in every block, these counts have the same distribution, and we use $A_3$ to denote a generic count. 
Let $u$ be the time offset of the first level-3 service slot within a level-2 block, so the batch from block $b$ enters the queue at $(b+1)T_3+u$. Let $Q_3(b)$ count the matchings already present just before this addition. Between this addition and the next, there are $S_3=L_3a_2+\varepsilon T_3$ level-3 service slots and no other additions. Therefore, starting from an empty switch, the queue evolves as
\begin{align}\label{eq:matching-recursion}
    Q_3(0)=0,\qquad
    Q_3(b+1)=\big(Q_3(b)+A_3(b)-S_3\big)^+.
\end{align}
We will analyze the maximum queue length $q_3=\sup_{b\geq 0}\E{Q_3(b)}$ in the queue of matchings. This is useful because of the following lemma, which relates the number of level-3 jobs to the number of matchings formed from each arrival block and the number already in the queue.

\begin{lemma}\label{lem:population-three}
Starting from an empty switch,
\begin{equation}\label{eq:level-three-counting}
 \sup_{t\ge0}\E{X_3(t)}
 \le4nL_3a_2+n\big(q_3+\E{A_3}\big).
\end{equation}
\end{lemma}
\begin{proof}
Fix a slot $t$. We count the level-3 jobs that are waiting to be added to the queue of matchings, as well as the level-3 jobs already in the queue. 

\par\addvspace{\medskipamount}
\noindent\emph{Jobs waiting to be added to the queue of matchings.}
These jobs necessarily got promoted to level-3 during either the current or preceding level-3 block, together containing at most $2L_3$ level-2 blocks. Each such job was promoted after its batch's last level-2 service slot and belongs to the corresponding $R_2$. Thus, counting the jobs in these graphs and applying Lemma~\ref{lem:population-two}(ii) gives an expected number of at most
\(
    (2L_3)\cdot(2na_2)=4nL_3a_2.
\)
Note that this bound also holds when the preceding block may not exist.

\par\addvspace{\medskipamount}
\noindent\emph{Jobs already in the queue of matchings.}
Denote $\ell=\lfloor t/T_3\rfloor$ as the index of the level-3 block that time slot $t$ belongs to. The queue is empty when $\ell=0$, so we assume without loss of generality that $\ell\ge1$. Before time $\ell T_3+u$, there have been no additions or level-3 service slots in the current block, and the queue therefore contains $Q_3(\ell-1)$ matchings. At that time, $A_3(\ell-1)$ matchings are added. Since there are no further additions during the block, the number of matchings can only decrease afterward. Finally, each matching contains at most $n$ jobs, so the expected number of jobs in this group is at most
\[
    n \, \E{Q_3(\ell-1)+A_3(\ell-1)}
    \leq n\big(q_3+\E{A_3}\big).
\]
This bound also holds when $\ell=0$, since the queue is empty.
Adding the two bounds and taking the supremum over $t$ proves~\eqref{eq:level-three-counting}.
\end{proof}

It remains to bound $\E{A_3}$ and $q_3$. We first bound the matching arrivals using the following lemma, proved in Section~\ref{sec:support-three}. 

\begin{lemma}[Matching arrivals]\label{lem:matching-three}
    Let $\delta_3=S_3-\E{A_3}$. Then
    \begin{enumerate}[itemsep=0.5em]
        \item[$\mathrm{(i)}$] $0\le\E{A_3}-L_3a_2\le\varepsilon T_3/2$.
        \item[$\mathrm{(ii)}$] $\Var{A_3}\le34T_3h$.
        \item[$\mathrm{(iii)}$] $\varepsilon T_3/2\le\delta_3\le\varepsilon T_3$.
    \end{enumerate}
\end{lemma}

Part (i) directly bounds $\E{A_3}$ in~\eqref{eq:level-three-counting}. To bound $q_3$, we use the following lemma.

\begin{lemma} \label{lem:queue-three}
    Starting from an empty switch,
    \begin{equation}\label{eq:matching-queue-bound}
         q_3
        \le\frac{34h}{\varepsilon}+\frac{\varepsilon T_3}{2}
        \le(34+4K^3)\frac h\varepsilon.
    \end{equation}
\end{lemma}
\begin{proof}
Since the switch starts empty, the matching queue is empty before the first addition, so $Q_3(0)=0$. Between two consecutive additions, it has $S_3$ level-3 service slots, all in block $b+1$, and receives no other matchings. Serving one matching in each such slot whenever the queue is nonempty gives~\eqref{eq:matching-recursion}.
Each $A_3(b)$ depends only on the arrivals and initial levels in level-3 block $b$, through the policy's fixed deterministic rules.
Thus the $A_3(b)$ are independent and identically distributed.
Moreover,
\[
 0\le A_3(b)\le\max_vN^{(v)}(T_3)\le nT_3.
\]
We use the following variant of Kingman's bound, which relates the expected queue length to the arrival variance and service slack. Its proof is in Section~\ref{sec:support-three}.

\begin{lemma}[Scalar queue]\label{lem:scalar-queue}
Let $A(m)$ be bounded iid nonnegative random variables with mean $\lambda$ and variance $v$. If $B>\lambda$, set $\delta=B-\lambda$. For
\[
    Q(0)=0,
    ~~~~~~
    Q(m+1)=\big(Q(m)+A(m)-B\big)^+,
\]
we have 
\(
    \sup_{m\geq0} \mathbb{E}\big[ Q(m) \big] \leq v/(2\delta)+\delta/2.
\)
\end{lemma}

Set $\delta_3=S_3-\E{A_3}>0$, where positivity follows from Lemma~\ref{lem:matching-three}(iii). Applying Lemma~\ref{lem:scalar-queue} with $\lambda=\E{A_3}$, $v=\Var{A_3}$, and $B=S_3$, and then the estimates in Lemma~\ref{lem:matching-three}(ii)--(iii), we have
\[
     q_3=\sup_{b\ge0}\E{Q_3(b)}
     \leq \frac{\Var{A_3}}{2\delta_3}+\frac{\delta_3}{2}
     \leq \frac{34T_3h}{\varepsilon T_3}+\frac{\varepsilon T_3}{2}
      =\frac{34h}{\varepsilon}+\frac{\varepsilon T_3}{2}
     \leq (34+4K^3)\frac h\varepsilon.
\]
The last inequality uses $T_3\le8K^3h\varepsilon^{-2}$.
\end{proof}
Substituting the bounds on $\E{A_3}$ and $q_3$ into Lemma~\ref{lem:population-three} gives the desired bound
\begin{equation}\label{eq:level-three-population-bound}
     \sup_{t\geq 0}\E{X_3(t)}
     \leq n \, \big( 5 L_3a_2 + 34h/\varepsilon+\varepsilon T_3\big).
\end{equation}

\subsection{Completing the three-level proof}\label{sec:completion-three}

Combining Lemmas~\ref{lem:queue-length-one} and~\ref{lem:queue-length-two} with~\eqref{eq:level-three-population-bound} in the decomposition~\eqref{eq:stage-decomposition} gives
\begin{equation}\label{eq:three-total-bound}
    \begin{aligned}
     \E{Q_\Sigma(t)}
     &\leq n \big( \, 2\rho T_1+4L_2a_1+5L_3a_2+34h/\varepsilon+\varepsilon T_3 \, \big)
     \\
     &\leq n \big( \, 2T_1 + 4 L_2a_1 + 5 L_3a_2 +(34+8K^3)h/\varepsilon \, \big) \, .
    \end{aligned}
\end{equation}
Each of these terms can be bounded as follows. Using~\eqref{eq:three-blocks},~\eqref{eq:three-block-scales}, and the fact that $a_j<(C+1)\sqrt{T_jh}$, we have
\[
\begin{aligned}
    T_1&\le2Kh\varepsilon^{-8/7},\\
    L_2a_1
    &\le(C+1)L_2\sqrt{T_1h}
     \le(C+1)(2K)^{3/2}h\varepsilon^{-8/7},\\
    L_3a_2
    &\le(C+1)L_3\sqrt{T_2h}
     \le(C+1)(2K)^2h\varepsilon^{-8/7}.
\end{aligned}
\]
Since $C,K$ are fixed and $h/\varepsilon\le h\varepsilon^{-8/7}$, the bound is $O(nh\varepsilon^{-8/7})$, uniformly in $t$.
Finally,
\begin{equation}\label{eq:log-factor}
 h=\log K+\log(n/\varepsilon)
   =\Theta\!\left(\log(n/\varepsilon)\right).
\end{equation}
This proves Theorem~\ref{thm:main} for $k=3$.

\section{Proofs of supporting lemmas for Section~\texorpdfstring{\ref{sec:proof-three}}{}}\label{sec:supporting-proofs}
We present the supporting proofs for Section~\ref{sec:proof-three}, in the same order as the three levels. We first state two lemmas that are used across multiple levels. The first lemma is the key ingredient to show that having our jobs start at a randomized level ensures that $s$-factors exist at all levels $j<k$ with high probability. The second lemma allows us to control the fluctuations of leftover jobs at their respective ports. The proofs of these lemmas are in Appendix~\ref{sec:technical-proofs}.

\begin{lemma} \label{lem:randomprotected}
Let $n\ge2$, let $T$ be a positive integer, and let $q\in[0,1]$. Let $H$ be a bipartite multi-graph with parts $\mathcal I,\mathcal O$, each of size $n$, and independent edge multiplicities $\Binomial(T,q)$. Suppose $nTq\ge d/2$ and $d\ge192\log n$. With probability at least $1-4e^{-d/96}$, it holds for every $X\subseteq\mathcal I$ and $Y\subseteq\mathcal O$ that
\begin{equation}\label{eq:protectedcut}
    e_H(\mathcal I\setminus X,Y)
    \geq \frac{d(n-|X|)|Y|}{4n}
    \quad\text{whenever }0<|X|-|Y|<\min\{|Y|,n-|X|\} \, .
\end{equation}
On this event, for every $b$ with $d\ge24b>0$, every bipartite multi-graph $R$ such that $H\subseteq R$ and
\[
    d-b\leq\deg_R(v)\leq d+b
    \quad\text{for every }v\in\mathcal I\cup\mathcal O
\]
contains an $s$-factor for every nonnegative integer $s\le d-3b$. 
\end{lemma}

\begin{lemma} \label{lem:matching-moments}
    Let $D$ be a fixed positive constant, and suppose the integer $T$ satisfies
    \begin{equation}\label{eq:terminal-length}
        K^2h\varepsilon^{-2}\le T\le Dh\varepsilon^{-2}.
    \end{equation}
    For a block with $T$ slots, let $N^{(v)}(T)$ count its arrivals at port $v$, and set $Z_T=\max_{v\in\mathcal I\cup\mathcal O}N^{(v)}(T)-\rho T$.
    Suppose $c \geq 0$ is deterministic, and $F$ is a random variable with $0\leq F \leq T$ with $\E{F} \leq h$. 
    Define the quantities
    \[
        A=c+Z_T+F,\qquad B=c+\varepsilon T,\qquad\delta=B-\E{A} .
    \]
    For all sufficiently large $n$,
\begin{itemize}
    \item[$\mathrm{(i)}$] $0\le\E{A}-c\le\varepsilon T/2$.
    \item[$\mathrm{(ii)}$] $\Var{A}\le34Th$.
    \item[$\mathrm{(iii)}$] $\varepsilon T/2\le\delta\le\varepsilon T$.
\end{itemize}
\end{lemma}

\subsection{Proofs for level-1 jobs}\label{sec:support-one}
\begin{proof}[Proof of Lemma~\ref{lem:success-one}]
(i) Recall that $C=108$, $K=144$, and that $a_1$ satisfies
\[
    C\sqrt{T_1h}\leq a_1<(C+1)\sqrt{T_1h}.
\]
Consider a level-1 arrival block, and let $G_1$ be the graph of its jobs that started at level 1. Its edge multiplicities are independent $\Binomial(T_1,\rho(1-\theta)/n)$ variables. Thus, the mean degree at any port is
\[
    \mu_1=\rho(1-\theta)T_1=\rho T_1-a_1/2.
\]
Define the degree and cut events
\[
    \begin{aligned}
        \mathcal D_1
        &=\bigcap_{v\in\mathcal I\cup\mathcal O}
          \bigl\{|\deg_{G_1}(v)-\mu_1|\leq a_1/9\bigr\},\\[4pt]
        \mathcal C_1
        &=\bigcap_{\substack{X\subseteq\mathcal I,\ Y\subseteq\mathcal O\\
          0<|X|-|Y|<\min\{|Y|,n-|X|\}}}
          \left\{e_{G_1}(\mathcal I\setminus X,Y)
          \geq\frac{\mu_1(n-|X|)|Y|}{4n}\right\}.
    \end{aligned}
\]
Suppose for the moment that $\mu_1\geq24(a_1/9)$ and $\mu_1\geq192\log n$; we verify these inequalities below. On $\mathcal D_1\cap\mathcal C_1$, Lemma~\ref{lem:randomprotected} applies with $H=R=G_1$ and $(d,b)=(\mu_1,a_1/9)$. It yields that an $S_1$-factor exists in $G_1$ with high probability, because
\[
    d-3b=\mu_1-a_1/3
    =\rho T_1-5a_1/6
    =S_1+a_1/6\geq S_1.
\]
More specifically, the standard Chernoff bound (Appendix~\ref{sec:chernoff}), combined with a union bound over the $2n$ ports and Lemma~\ref{lem:randomprotected}, yields
\[
    \Pp(\mathcal F^{(1)})
    \leq\Pp(\mathcal D_1^c)+\Pp(\mathcal C_1^c)
    \leq4n\exp\!\left(-\frac{a_1^2}{243\mu_1}\right)
         +4e^{-\mu_1/96}.
\]
Here the Chernoff bound~\eqref{eq:chernoff} is applied with $x=a_1/9\leq\mu_1$.

To complete the proof of (i), it remains to verify the two inequalities assumed above and bound the right-hand side. By~\eqref{eq:three-block-scales},
\[
    \frac{\mu_1}{a_1}
    =\frac{\rho T_1}{a_1}-\frac12
    \geq\frac{\rho\sqrt{T_1/h}}{C+1}-\frac12
    \geq\frac{12\rho}{109\varepsilon^{4/7}}-\frac12
    \to\infty.
\]
In particular, $a_1\leq\rho T_1$ for sufficiently large $n$, so
\[
    \frac{\mu_1}{h}
    \geq\frac{\rho T_1}{2h}
    \geq72\rho\varepsilon^{-8/7}
    \to\infty.
\]
This verifies that $\mu_1\geq24(a_1/9)$ and $\mu_1\geq96\cdot48h\geq192\log n$ for sufficiently large $n$. Since $a_1^2\geq C^2T_1h$, $\mu_1\leq T_1$, and $C^2/243=48$, we obtain
\[
    \Pp(\mathcal F^{(1)})
    \leq4ne^{-48h}+4e^{-48h}
    \leq8e^{-47h}\leq e^{-40h},
\]
where the last two inequalities use $n\leq e^h$ and $h\geq1$.

\par\addvspace{\medskipamount}
\noindent (ii) We now show that $\E{F_1} \leq T_1 e^{-40h} \leq a_1.$ The first inequality uses that $S_1\le T_1$, and so
\[
    \E{F_1}=S_1 \cdot \Pp(\mathcal F^{(1)})\le T_1e^{-40h} .
\]
For the second inequality, we use~\eqref{eq:three-block-scales}, $a_1\ge C\sqrt{T_1h}$, and $e^{-40h}\le e^{-2h}=(\varepsilon/(Kn))^2$ to get
\[
     \frac{T_1e^{-40h}}{a_1}
     \le\frac{\sqrt{2K}}{C}\varepsilon^{-4/7}e^{-40h}
     \le\frac{\sqrt{2K}}{CK^2n^2}\varepsilon^{10/7}\le1.
\]
\end{proof}

\begin{proof}[Proof of Lemma~\ref{lem:population-one}]
Taking expectations in~\eqref{eq:leftovers-one}, with $\E{N^{(v)}(T_1)}=\rho T_1$, and applying Lemma~\ref{lem:success-one} yields
\[
    \mathbb{E} \big[ D_1^{(v)} \big] = a_1+\E{F_1}\leq 2a_1 \, .
\]
This proves (i). Summing over the input ports gives $\E{J_1}\le2na_1$, proving (ii).
\end{proof}

\subsection{Proofs for level-2 jobs}\label{sec:support-two}
\begin{proof}[Proof of Lemma~\ref{lem:success-two}]

(i) Fix one level-2 arrival block, and let $\mathcal J$ index its $L_2$ constituent level-1 blocks. Define the auxiliary event
\[
    \mathcal G^{(2)}=\mathcal F^{(2)}
    \cup
    \bigg( \, \bigcup_{c\in\mathcal J}\mathcal F_c^{(1)} \, \bigg) \, .
\]
Let $R=G_2$ be the graph collected for this block's level-2 service, and let $N^{(v)}$ count its original arrivals at port $v$. Then $N^{(v)}\sim\Binomial(nT_2,\rho/n)$. Set $d_2=L_2a_1$. Summing~\eqref{eq:leftovers-one} over the constituent level-1 blocks gives
\[
 \deg_R(v)=N^{(v)}-\rho T_2+d_2+\sum_{c\in\mathcal J}F_{1,c}.
\]
If level-1 service succeeds for jobs from every block in $\mathcal J$, then exactly $L_2S_1$ jobs have been served at every port, and hence
\begin{equation}\label{eq:three-residual-degree-success}
    \deg_R(v)=N^{(v)}-L_2S_1=d_2+(N^{(v)}-\rho T_2).
\end{equation}
Let $H_2$ be the graph of this block's jobs that started at level 2. These jobs skip level 1, so $H_2\subseteq R$ on every outcome. Its edge multiplicities are independent $\Binomial(T_2,\rho\theta/n)$ variables. Thus, the mean degree at any port is
\[
    \E{\deg_{H_2}(v)}
    =\rho T_2\theta
    =\frac{a_1T_2}{2T_1}
    =\frac{d_2}{2}.
\]
Define the degree and cut events
\[
    \begin{aligned}
        \mathcal D_2
        &=\bigcap_{v\in\mathcal I\cup\mathcal O}
          \bigl\{|N^{(v)}-\rho T_2|\leq a_2/9\bigr\},\\[4pt]
        \mathcal C_2
        &=\bigcap_{\substack{X\subseteq\mathcal I,\ Y\subseteq\mathcal O\\
          0<|X|-|Y|<\min\{|Y|,n-|X|\}}}
          \left\{e_{H_2}(\mathcal I\setminus X,Y)
          \geq\frac{d_2(n-|X|)|Y|}{4n}\right\}.
    \end{aligned}
\]
Suppose for the moment that $d_2\geq24(a_2/9)$, $d_2\geq192\log n$, and $a_2/9\leq\rho T_2$; we verify these inequalities below. If no constituent level-1 service fails and $\mathcal D_2\cap\mathcal C_2$ holds, the factor conclusion of Lemma~\ref{lem:randomprotected} gives an $S_2$-factor because
\[
    d_2-3(a_2/9)=d_2-a_2/3
    =S_2+2a_2/3\geq S_2.
\]
Therefore,
\(
    \mathcal G^{(2)}\subseteq
    \big(\bigcup_{c\in\mathcal J}\mathcal F_c^{(1)}\big)
    \cup\mathcal D_2^c\cup\mathcal C_2^c.
\)
The Chernoff bound (Appendix~\ref{sec:chernoff}), combined with a union bound and Lemmas~\ref{lem:success-one} and~\ref{lem:randomprotected}, yields
\[
        \Pp(\mathcal G^{(2)})
        \leq\sum_{c\in\mathcal J}\Pp(\mathcal F_c^{(1)})
             +\Pp(\mathcal D_2^c)+\Pp(\mathcal C_2^c)
        \leq L_2e^{-40h}
             +4n\exp\!\left(-\frac{a_2^2}{243\rho T_2}\right)
             +4e^{-d_2/96}.
\]
Here the Chernoff bound~\eqref{eq:chernoff} is applied with $x=a_2/9\leq\rho T_2$ at each of the $2n$ ports.

It remains to verify the three inequalities assumed above and bound the right-hand side. By~\eqref{eq:three-block-scales} and the bounds on $a_j$, we have that
\begin{equation}\label{eq:three-factor-parameters}
    \begin{aligned}
        \frac{d_2}{a_2}
        &\geq\frac{C\sqrt{L_2}}{C+1}
         \geq\frac{1296}{109}\varepsilon^{-2/7}>8,\\
        \frac{d_2}{96h}
        &\geq\frac{CL_2\sqrt{T_1/h}}{96}
         \geq1944\varepsilon^{-8/7}\geq48,\\
        \frac{a_2}{9\rho T_2}
        &\leq\frac{C+1}{9\rho}\sqrt{\frac h{T_2}}
         \leq\frac{109\varepsilon^{6/7}}{1296\rho}
         \to0.
    \end{aligned}
\end{equation}
Thus $d_2\geq24(a_2/9)$, $d_2\geq96\cdot48h\geq192\log n$, and $a_2/9\leq\rho T_2$ for sufficiently large $n$. Since $a_2^2\geq C^2T_2h$ and $\rho\leq1$, we obtain
\[
        \Pp(\mathcal G^{(2)})
        \leq L_2e^{-40h}+4ne^{-48h}+4e^{-48h}
        \leq L_2e^{-40h}+8e^{-47h}
        \leq(L_2+1)e^{-40h}.
\]
Since $\mathcal F^{(2)}\subseteq\mathcal G^{(2)}$, this proves the failure bound.

\par\addvspace{\medskipamount}
\noindent (ii) We show that $\E{F_2} \leq T_2(L_2+1)e^{-40h}\le a_2$. Since $L_2S_1+S_2=\rho T_2-a_2\le T_2$, Lemma~\ref{lem:success-one} and the failure bound above yield
\[
\begin{aligned}
    \E{F_2}
    =S_1\sum_{c=1}^{L_2}\Pp(\mathcal F_c^{(1)})
          +S_2 \cdot \Pp(\mathcal F^{(2)})
     \, \le \, 
        \big[ L_2 S_1 + S_2(L_2+1) \big]e^{-40h}
     \, \le \,  T_2(L_2+1)e^{-40h}.
\end{aligned}
\]
To bound the error term, note that
\[
    e^{-40h}\le e^{-2h}
    =\frac{\varepsilon^2}{K^2n^2}
    \le\frac{\varepsilon^2}{K^2},
    \qquad
    L_2+1\le2L_2\le4K\varepsilon^{-4/7}.
\]
By~\eqref{eq:three-block-scales} and $a_2\ge C\sqrt{T_2h}$,
\[
    \frac{T_2(L_2+1)e^{-40h}}{a_2}
    \le\frac{8K^2}{C}\varepsilon^{-10/7}e^{-40h}
     \le\frac8C\varepsilon^{4/7}\le1.
\]
The last inequality uses $C=108$ and $0<\varepsilon<1$. This concludes the proof.
\end{proof}

\begin{proof}[Proof of Lemma~\ref{lem:population-two}]
Taking expectations in~\eqref{eq:leftovers-two}, with $\E{N^{(v)}(T_2)}=\rho T_2$, and applying Lemma~\ref{lem:success-two} gives $\mathbb{E} \big[ D_2^{(v)} \big] = a_2+\E{F_2} \leq 2a_2$, proving (i). Summing over the input ports gives $\E{J_2}\le2na_2$, proving (ii).
\end{proof}

\subsection{Proofs for level-3 jobs}\label{sec:support-three}
\begin{proof}[Proof of Lemma~\ref{lem:matching-three}]
For one original level-3 block, let $D_3^{(v)}$ count its jobs at port $v$ that reach level 3, and set
\[
 Z_3=\max_{v\in\mathcal I\cup\mathcal O}N^{(v)}(T_3)-\rho T_3,
 \qquad F_3=\sum_{c=1}^{L_3}F_{2,c}.
\]
Sum the residual-degree identity~\eqref{eq:leftovers-two} over the $L_3$ constituent level-2 blocks. Since $T_3=L_3T_2$ and $F_3=\sum_{c=1}^{L_3}F_{2,c}$,
\begin{equation}\label{eq:three-terminal-residual}
\begin{aligned}
    D_3^{(v)}
    =N^{(v)}(T_3)-L_3(\rho T_2-a_2)+F_3
    =N^{(v)}(T_3)-\rho T_3+L_3a_2+F_3.
\end{aligned}
\end{equation}
By K\H{o}nig's edge-coloring theorem, the number of matchings is the maximum degree. The same scalar $F_3$ appears at every port, so
\begin{equation}\label{eq:matching-arrivals}
 A_3=\max_vD_3^{(v)}=L_3a_2+Z_3+F_3.
\end{equation}
Since these batches reserve $L_3(L_2S_1+S_2)\le T_3$ service slots in total, $0\le F_3\le T_3$. Lemma~\ref{lem:success-two} gives
\begin{equation}\label{eq:failure-moments}
 \E{F_3}\le L_3T_2(L_2+1)e^{-40h}=T_3(L_2+1)e^{-40h}.
\end{equation}
To bound the last expression, use $L_2+1\le2L_2\le4K\varepsilon^{-4/7}$, $T_3\le8K^3h\varepsilon^{-2}$, and $e^{-40h}\le e^{-6h}=(\varepsilon/(Kn))^6$. Then
\[
 T_3(L_2+1)e^{-40h}
 \le32K^4h\varepsilon^{-18/7}e^{-6h}
 =\frac{32\varepsilon^{24/7}}{K^2n^6}\,h\le h.
\]
The last inequality uses $K=144$, $0<\varepsilon<1$, and $n\ge1$.

Lemma~\ref{lem:matching-moments} applies to~\eqref{eq:matching-arrivals} with $T=T_3$, $D=(2K)^3$, $c=L_3a_2$, and $F=F_3$.
To check its hypotheses,~\eqref{eq:three-block-scales} gives~\eqref{eq:terminal-length}, and the preceding bounds give $0\le F_3\le T_3$ and $\E{F_3}\le h$.

\par\addvspace{\medskipamount}
\noindent (i)--(ii)
The proof follows directly from parts (i) and (ii) of Lemma~\ref{lem:matching-moments}.

\par\addvspace{\medskipamount}
\noindent (iii)
Since $S_3=L_3a_2+\varepsilon T_3$, part (i) gives
\[
 \delta_3=S_3-\E{A_3}
 =\varepsilon T_3-\bigl(\E{A_3}-L_3a_2\bigr)
 \in[\varepsilon T_3/2,\varepsilon T_3].\qedhere
\]
\end{proof}

\begin{proof}[Proof of Lemma~\ref{lem:scalar-queue}]
For any queue length $q\geq0$, independence of $A(m)$ from $Q(m)$ gives
\begin{align*}
    \mathbb{E}\big[ Q(m+1)^2-Q(m)^2\mid Q(m)=q \big] 
    &= \mathbb{E} \big[ \big((q+A(m)-B)^+\big)^2-q^2 \big]
    \\
    &
    \leq
    \mathbb{E} \big[ (q+A(m)-B)^2-q^2 \big]  \\
    &=2q(\lambda-B)+v+(\lambda-B)^2
    = -2\delta q+v+\delta^2.
\end{align*}
Let $\overline Q(m)$ be a queue obeying the same recursion, initialized in stationarity independently of future arrivals. Bounded arrivals and $\lambda<B$ ensure that this stationary distribution exists and has finite second moment.\footnote{To see this, start at $Q(0)=0$ and let $\tau=\inf\{m\geq1:Q(m)=0\}$. Since $X(m)=A(m)-B$ is bounded and $\E{X(m)}=-\delta<0$, some $\eta>0$ satisfies $r:=\E{e^{\eta X(m)}}<1$. On $\{\tau>\ell\}$, the queue stays positive through time $\ell$, so $Q(\ell)=\sum_{m=0}^{\ell-1}X(m)$. Hence, for $\ell\geq1$,
\[
    \Pp(\tau>\ell)
    \leq\Pp\!\left(\sum_{m=0}^{\ell-1}X(m)>0\right)
    \leq r^\ell.
\]
Thus $\E{\tau^3}<\infty$. If $A(m)\leq M$, then $Q(m)\leq Mm$. The stationary distribution obtained from the iid cycles between returns to zero therefore satisfies
\[
    \E{\overline Q(0)^2}
    =\frac{\E{\sum_{m=0}^{\tau-1}Q(m)^2}}{\E{\tau}}
    \leq\frac{M^2\E{\tau^3}}{\E{\tau}}<\infty.
\]
}
Applying the drift inequality to $\overline Q(m)$ and taking expectations gives
\[
    0=\E{\overline Q(m+1)^2-\overline Q(m)^2}
    \leq -2\delta \, \E{\overline Q(m)}+v+\delta^2,
\]
where the equality follows from stationarity. Thus $\E{\overline Q(m)}\leq v/(2\delta)+\delta/2$.

Finally, drive $Q(m)$ and $\overline Q(m)$ by the same arrivals. Since $Q(0)=0\leq\overline Q(0)$ and the recursion is increasing in the queue length, $Q(m)\leq\overline Q(m)$ for every $m$. Therefore,
\[
    \sup_{m\geq 0} \, \E{Q(m)}
    \leq\E{\overline Q(0)}
    \leq\frac{v}{2\delta}+\frac{\delta}{2}. \qedhere
\]
\end{proof}

\section*{Acknowledgments}
Taha Ameen and Weina Wang are supported in part by U.S.\ National Science Foundation (NSF) grants ECCS-2145713, CCF-2403194, CCF-2428569, and ECCS-2432545.

The authors thank Richard Peng for initial discussions on the policy design and for checking some preliminary ideas with GPT.

\bibliographystyle{alpha}
\bibliography{references}

\appendix

\section{Proofs of technical lemmas}\label{sec:technical-proofs}

We prove the technical lemmas stated in Section~\ref{sec:supporting-proofs}. 

\subsection{Chernoff bounds}\label{sec:chernoff}
The following Chernoff bounds are used in the proof of Theorem~\ref{thm:main}.

\begin{lemma} \label{lem:chernoff}
    Let $Z$ be a sum of independent Bernoulli random variables, with mean $\mu=\E{Z}>0$. For $0<x\le\mu$, the following bounds hold:
    \begin{enumerate}
    \item[(i)] \emph{Lower tail}~\cite[Theorem~4.5(2)]{mitzenmacher2017probability}:
    \begin{equation}\label{eq:chernoff-lower}
         \Pp(Z\le\mu-x)\le\exp\!\left(-\frac{x^2}{2\mu}\right).
    \end{equation}
    
    \item[(ii)] \emph{Two-sided deviation}~\cite[Corollary~4.6]{mitzenmacher2017probability}:
    \begin{equation}\label{eq:chernoff}
        \Pp(|Z-\mu|>x)\le2\exp\!\left(-\frac{x^2}{3\mu}\right).
    \end{equation}
    \end{enumerate}
\end{lemma}

\subsection{Factors and decompositions into matchings}\label{sec:graph-tools}\label{sec:randomprotected-proof}

In this section, we give the deterministic graph results used to prove Lemma~\ref{lem:randomprotected} and to decompose the policy's graphs into matchings. The following cut condition is necessary and sufficient for a bipartite multi-graph to contain an $s$-factor.

\begin{lemma}[Cut criterion]\label{lem:cut}
    Let $R$ be a bipartite multi-graph with $|\mathcal I|=|\mathcal O|=n$, and let $s$ be a nonnegative integer. Then $R$ contains an $s$-factor if and only if
    \begin{equation}\label{eq:cutcriterion}
        e_R(X,\mathcal O\setminus Y)\geq s(|X|-|Y|)
        \quad
        \text{for every } X\subseteq\mathcal I,\ Y\subseteq\mathcal O .
    \end{equation}
\end{lemma}

\begin{proof}
We use a max-flow min-cut argument. Construct a flow network by connecting a source to every input with capacity $s$, connecting each input $i\in\mathcal I$ to each output $j\in\mathcal O$ with capacity $e_R(\{i\},\{j\})$, and connecting every output to a sink with capacity $s$. A flow of value $sn$ saturates every arc leaving the source and every arc entering the sink. Hence, an integral flow of value $sn$ corresponds to an $s$-factor of $R$, where the flow on the arc from $i$ to $j$ specifies how many of the parallel edges between $i$ and $j$ are selected.

Consider any source--sink cut whose source side contains the inputs in $X\subseteq\mathcal I$ and the outputs in $Y\subseteq\mathcal O$. Counting the arcs from the source to $\mathcal I\setminus X$, from $X$ to $\mathcal O\setminus Y$, and from $Y$ to the sink, its capacity is
\begin{equation*}
    s\bigl(n-|X|\bigr)+e_R(X,\mathcal O\setminus Y)+s|Y|.
\end{equation*}
Thus, by the max-flow min-cut theorem and the integrality of maximum flows, $R$ contains an $s$-factor if and only if
\[
    s\big(n-|X|\big)+e_R(X,\mathcal O\setminus Y)+s|Y|\geq sn
\]
or equivalently, if and only if
\(
    e_R(X,\mathcal O\setminus Y)\geq s\big(|X|-|Y| \big)
\)
for every $X\subseteq\mathcal I$ and $Y\subseteq\mathcal O$.
\end{proof}
The next lemma justifies serving an $s$-factor in $s$ slots. It also determines the number of matchings added to the queue at the last level.

\begin{lemma}[Decomposition into matchings]\label{lem:decomposition}
    Let $R$ be a bipartite multi-graph with $|\mathcal I|=|\mathcal O|=n$. 
    \begin{enumerate}
        \item[$\mathrm{(i)}$] If $R$ is $s$-regular, it decomposes into $s$ perfect matchings.
        \item[$\mathrm{(ii)}$] If $R$ has maximum degree $M$, it decomposes into $M$ matchings.
    \end{enumerate} 
\end{lemma}

\begin{proof}
(i) Suppose first that $R$ is $s$-regular with $s>0$. For $X\subseteq\mathcal I$, let $N_R(X)$ be its set of neighbors. Then
\[
    s|X|= e_R\big( X,N_R(X) \big) \leq s|N_R(X)|,
\]
so Hall's theorem gives a perfect matching. Removing this perfect matching leaves an $(s-1)$-regular graph, and iterating gives $s$ perfect matchings. The case $s=0$ is immediate.

\par\addvspace{\medskipamount}
\noindent (ii) For a graph of maximum degree $M$, the total degree deficits on the two sides are the same:
\[
    \sum_{i\in\mathcal I}(M-\deg_R(i))
    =n M - |E(R)|
    =\sum_{j\in\mathcal O}(M-\deg_R(j)).
\]
We may therefore add an edge between an input and an output with positive deficits until all deficits are zero. The resulting graph is $M$-regular, so, by (i), it decomposes into $M$ perfect matchings. Deleting the added edges from these matchings gives the required decomposition of $R$.
\end{proof}

The next lemma gives sufficient conditions for the cut criterion using degree bounds on $R$ and edge bounds on a subgraph $H\subseteq R$. Jobs that start at the level being analyzed provide the subgraph $H$; the conclusion gives an $s$-factor of the full graph $R$.

\begin{lemma} \label{lem:protected}
    Let $H$ and $R$ be bipartite multi-graphs with parts $\mathcal I,\mathcal O$, each of size $n$. Let $d$ and $b$ satisfy $d\geq 24b > 0$. Assume that $H\subseteq R$ and
    \[
        d-b \leq \deg_R(v) \leq d+b
        \quad\text{for every } v\in\mathcal I\cup\mathcal O \, .
    \]
    For $X\subseteq\mathcal I$ and $Y\subseteq\mathcal O$, assume also that
    \begin{equation}\label{eq:deterministic-protectedcut}
        e_H \big(\mathcal I\setminus X,Y \big)
        \geq \frac{d(n-|X|) |Y|}{4n}
        \quad\text{whenever }  0 < |X| - |Y| <\min\{|Y|,n-|X|\}.
    \end{equation}
    Then $R$ contains an $s$-factor for every nonnegative integer $s\leq d-3b$.
\end{lemma}

\begin{proof}
Let $x=|X|$ and $y=|Y|$. By Lemma~\ref{lem:cut}, it suffices to prove
\begin{align} \label{eq:cut-ineq-specific}
    e_R(X,\mathcal O\setminus Y) \geq (d-3b)(x-y)
    \quad\text{for every }X\subseteq\mathcal I,\ Y\subseteq\mathcal O.
\end{align}
Suppose this fails for some $X,Y$; necessarily $x>y$. We may take $x+y\le n$ without loss of generality (otherwise, we can exchange the input and output labels and replace $(X,Y)$ by $(\mathcal O\setminus Y,\mathcal I\setminus X)$.
Both edge counts $e_R(X,\mathcal O\setminus Y)$ and $e_H(\mathcal I\setminus X,Y)$, and condition~\eqref{eq:deterministic-protectedcut}, are unchanged by this transformation).
We obtain
\begin{align}
    0\leq e_H(\mathcal I\setminus X,Y)
     \leq e_R(\mathcal I\setminus X,Y)
    &= e_R(X,\mathcal O\setminus Y)
      -\sum_{i\in X}\deg_R(i)+\sum_{j\in Y}\deg_R(j) \notag \\
    &\stackrel{\mathrm{(a)}}<
      (d-3b)(x-y)-(d-b)x+(d+b)y
      =2b(2y-x), \label{eq:cut-upper-bound}
\end{align}
where (a) uses the degree bounds, and the assumption that~\eqref{eq:cut-ineq-specific} does not hold. Thus,
\[
    x<2y,\qquad
    0<x-y<y \leq n-x,\qquad
    3x<2(x+y)\leq 2n.
\]
Consequently,
\[
    e_H(\mathcal I\setminus X,Y) \overset{\mathrm{(b)}}\ge\frac{d(n-x)y}{4n} \ge\frac{dy}{12} \overset{\mathrm{(c)}}\ge2by
      >2b(2y-x),
\]
where (b) is \eqref{eq:deterministic-protectedcut}, (c) uses $d\ge24b$, and the last inequality uses $x>y>0$. This contradicts~\eqref{eq:cut-upper-bound}.
\end{proof}

We now prove Lemma~\ref{lem:randomprotected}. 

\begin{proof}[Proof of Lemma~\ref{lem:randomprotected}]
For $X\subseteq\mathcal I$, $Y\subseteq\mathcal O$, with $x = |X|$ and $y = |Y|$, define $\mathcal B_{X,Y}$ as the event that $X$ and $Y$ form a \emph{bad cut}, i.e., 
\[
 \mathcal B_{X,Y}
 =\Big\{e_H(\mathcal I\setminus X,Y)
                 <\frac{d(n-x)y}{4n}\Big\}.
\]
Let $\mathcal B_{\le}$ (resp. $\mathcal B_{>}$) denote the event that there is such a bad cut with $0<x-y<\min\{y,n-x\}$ and $x+y\le n$ (resp. $x+y>n$). We bound the two events below.

\par\addvspace{\medskipamount}
\noindent \textit{Case $1$: $x+y\le n$.} Note that the distribution of the number of edges (including multiplicities) between $\mathcal{I}\setminus X$ and $Y$ satisfies
\[
    e_H(\mathcal I\setminus X,Y)
    \sim\Binomial\big(T(n-x)y,q\big) \, ,
    \quad
    \mathbb{E} \big[ e_H(\mathcal I\setminus X,Y) \big]
    =Tq(n-x)y\ge\frac{d(n-x)y}{2n}.
\]
Therefore, by~\eqref{eq:chernoff-lower},
\[
    \Pp(\mathcal B_{X,Y})
    \leq
    \exp\left(-\frac{d(n-x)y}{16n}\right)
    \stackrel{ \mathrm{(a)} }{\leq}
    \exp \big( -dy/48 \big) \, ,
\]
where (a) used that $n-x \geq n/3$. To see why this holds, note that the condition $0 < x - y < \min\{y,n-x\}$ implies that $x < 2y$. Further, since $x+y \leq n$, we have
\(
   3 x < 2(x+y) \leq 2n,  
\)
or equivalently $n-x \geq n/3$.
Since $x,y$ are integers, $y<x<2y$ implies $y\ge2$ and $x+y\le3y-1$. Taking a union bound,
\[
\begin{aligned}
 \Pp(\mathcal B_{\le})
 \leq \sum_{y=2}^{n}  \
       \sum_{\substack{y<x<2y\\x+y\leq n}}
       \binom nx\binom ny e^{-dy/48}
    & \stackrel{\mathrm{(b)}}{\leq}
       \sum_{y=2}^{n}y n^{3y-1}e^{-dy/48}\\
 &\leq \sum_{y\ge2}\exp\big(-(d/48-3\log n)y\big)\\
 &\stackrel{\mathrm{(c)}}{\leq}
       \sum_{y\ge2}e^{-dy/192}
       =\frac{e^{-d/96}}{1-e^{-d/192}}
       \leq 2e^{-d/96} \, .
\end{aligned}
\]
Here (b) uses the fact that $\binom nx\binom ny\leq n^{x+y}\le n^{3y-1}$ for all $x$ in the sum, and that there are at most $y$ choices of $x$. The next step uses $y\leq n$. For (c), $d\geq 192\log n$ yields $d/48-3\log n \geq d/192$. It also yields $e^{-d/192}\leq n^{-1}\leq 1/2$, proving the last inequality.

\par\addvspace{\medskipamount}
\noindent \textit{Case $2$: $x+y>n$.}
Exchange the two sides and replace $(X,Y)$ by $(\mathcal O\setminus Y,\mathcal I\setminus X)$. The transformed sets have total size $(n-y)+(n-x)=2n-x-y<n$, and the graph with its sides exchanged has the same independent binomial edge distribution. Thus
\[
    \Pp(\mathcal B_{>})\leq \Pp(\mathcal B_{\leq}) \leq 2e^{-d/96}.
\]
Finally,
\[
    \Pp\big(\eqref{eq:protectedcut}\text{ fails for some }X,Y\big)
    =\Pp(\mathcal B_{\leq}\cup\mathcal B_{>})
    \leq
    \Pp(\mathcal B_{\leq})+\Pp(\mathcal B_{>})
    \leq
    4e^{-d/96}.
\]
On the complementary event, Lemma~\ref{lem:protected} gives the stated factor conclusion for every $R$ containing $H$ with the required degree bounds.
\end{proof}

\subsection{Matching arrivals and service slack}\label{sec:matching-moments-proof}

To prove Lemma~\ref{lem:matching-moments}, we first bound the mean and second moment of the largest arrival fluctuation across ports.

\begin{lemma}[Maximum port fluctuation]\label{lem:max-fluctuation}
Let $D$ be a fixed positive constant, and suppose the integer $T$ satisfies
\[
 K^2h\varepsilon^{-2}\le T\le Dh\varepsilon^{-2}.
\]
For the arrivals in a $T$-slot block, let $N^{(v)}(T)$ count its arrivals at port $v$, and set $Z_T=\max_{v\in\mathcal I\cup\mathcal O}N^{(v)}(T)-\rho T$.
For all sufficiently large $n$,
\[
 0\le\E{Z_T}\le4\sqrt{Th},\qquad \E{Z_T^2}\le16Th.
\]
\end{lemma}

\begin{proof}
Set $V=\max_v|N^{(v)}(T)-\rho T|$ and $a=3\sqrt{Th}$.
Each $N^{(v)}(T)$ is binomial with mean $\rho T$.
Provided $a\le\rho T$, the two-sided Chernoff bound and a union bound over the $2n$ ports give
\[
    \Pp(V>a)
    \leq \sum_{v\in\mathcal I\cup\mathcal O}
          \Pp\bigl(|N^{(v)}(T)-\rho T|>a\bigr)
    \leq 4n\exp\!\left(-\frac{a^2}{3\rho T}\right)
    \leq 4ne^{-3h}.
\]
The required condition $a\le\rho T$ holds for all sufficiently large $n$, since
\(
    \frac{a}{\rho T}
    =\frac{3\sqrt h}{\rho\sqrt T}
    \le\frac{3\varepsilon}{\rho K}\to0.
\)
Also, $0\le N^{(v)}(T)\le nT$ implies $V\le nT$ on every outcome.
Splitting according to whether $V\le a$,
\begin{align*}
    \E{V^2}
    =\E{V^2 \cdot \ind_{\{V\le a\}}}
      +\E{V^2 \cdot \ind_{\{V>a\}}}
    & \leq a^2+(nT)^2 \cdot \Pp(V>a)
     \leq 9Th+4n^3T^2e^{-3h}\\
    &\leq \Big(9+\frac{4D\varepsilon}{K^3}\Big) \, Th
     \leq 16Th \, .
\end{align*}
Here the second-last inequality uses $T\le Dh\varepsilon^{-2}$ and $e^{-3h}=(\varepsilon/(Kn))^3$.
The last inequality holds for all sufficiently large $n$ because $D$ is fixed and $\varepsilon\to0$.
Finally, $|Z_T|\le V$ and $Z_T\ge N^{(v)}(T)-\rho T$ for each fixed $v$, so
\begin{align*}
    \E{Z_T^2}&\le\E{V^2}\le16Th,\\
    0=\E{N^{(v)}(T)-\rho T}
    &\le\E{Z_T}\le\sqrt{\E{Z_T^2}}\le4\sqrt{Th}.\qedhere
\end{align*}
\end{proof}

Combining this fluctuation bound with the contribution from failed tests gives the claimed mean, variance, and service slack.

\begin{proof}[Proof of Lemma~\ref{lem:matching-moments}]
\noindent (i)
Taking expectations in the definition of $A$ and using Lemma~\ref{lem:max-fluctuation} and $F\ge0$ gives
\[
\begin{aligned}
    0\le\E{A}-c
    =\E{Z_T}+\E{F}
    &\leq 4\sqrt{Th}+\E{F}
     \leq 4\sqrt{Th}+h\\
    &\le\left(\frac4K+\frac{\varepsilon}{K^2}\right)\varepsilon T
     \le\frac{\varepsilon T}{2}.
\end{aligned}
\]
The second-last inequality follows from $T \geq K^2h\varepsilon^{-2}$, and the last uses $K=144$ and $0<\varepsilon<1$.

\par\addvspace{\medskipamount}
\noindent (ii)
By the definition of $A$ and Lemma~\ref{lem:max-fluctuation},
\[
    \Var{A}
    \leq \E{(A-c)^2}
     =\E{(Z_T+F)^2}
    \leq 2\E{Z_T^2}+2\E{F^2}
    \leq 32Th + 2 \, \E{F^2}.
\]
Since $0\le F\le T$,
\[
    \E{F^2}\le T\E{F}\le Th.
\]
Substituting into the preceding bound gives $\Var{A}\leq 34Th$.

\par\addvspace{\medskipamount}
\noindent (iii)
By (i), we have
\[
    \delta
    =B-\E{A}
    =\varepsilon T-(\E{A}-c)
    \in[\varepsilon T/2,\varepsilon T].\qedhere
\]
\end{proof}

\section{Proof for general \texorpdfstring{$k$}{k}}\label{sec:extension}
Fix $k\ge2$ and consider the policy in Section~\ref{sec:policy-formal}, with $C=108$ and $K=144$. We again bound the number of jobs at each level, followed by the jobs from each arrival block that enter higher levels. All constants in $O_k(\cdot)$ may depend on $k$, but not on $n$, $\rho$, or time. The supporting proofs are in Section~\ref{sec:support-general}.

\subsection{Parameters and notation}\label{sec:parameters-general}
Set $\beta_j=2\alpha_k(1-2^{-j})$ for $1\le j\le k$.
Then $\beta_k=2$, and the block lengths satisfy
\begin{align}\label{eq:block-scales}
    K^jh\varepsilon^{-\beta_j}\le T_j\le(2K)^jh\varepsilon^{-\beta_j},
    \qquad
    T_j = \big( 1+o(1) \big) \, K^j h\varepsilon^{-\beta_j}.
\end{align}
Let $M$ denote a job's initial level. For $j<k$, let $\tau_j$ be the probability that it starts above level $j$, and let $m_j$ be the expected number of such arrivals per port in a level-$j$ block:
\begin{equation}\label{eq:future-reserves}
 \tau_j=\Pp(M>j)=\sum_{r=j}^{k-2}\theta_r,\qquad
 m_j=\rho T_j\tau_j=\frac{T_j}{2}\sum_{r=j}^{k-2}\frac{a_r}{T_r}.
\end{equation}
Empty sums are zero. Also define
\[
 d_1=\rho T_1,\qquad d_j=L_ja_{j-1}\ (2\le j<k),\qquad
 \mu_j=d_j-m_j,\qquad \nu_j=\rho T_j-m_j.
\]
Here $d_j$ is the mean arrival count per port minus the prescribed lower-level removals. The quantity $\mu_j$ also subtracts the mean $m_j$ of jobs starting above level $j$, whereas $\nu_j$ is the mean number of original arrivals with $M\le j$ at a port.

The service identities telescope to give
\begin{equation}\label{eq:planned-removals}
 \sum_{i=1}^{j-1}\frac{T_j}{T_i}S_i=\rho T_j-d_j,
 \qquad \sum_{i=1}^{j}\frac{T_j}{T_i}S_i=\rho T_j-a_j.
\end{equation}

\begin{lemma}[Parameters for the factor tests]\label{lem:parameters-general}
For every $1\le j<k$,
\begin{equation}\label{eq:reserve-bound}
 m_j\le5a_j/9\le2a_j/3.
\end{equation}
For all sufficiently large $n$, the initial-level probabilities are well-defined, the service counts are positive and feasible, and it holds that
\begin{gather}
 \mu_j\ge24(a_j/9),\qquad
 \mu_j-3(a_j/9)\ge S_j,\qquad
 \mu_j\ge96\cdot48h, \qquad
 a_j/9\le\nu_j\le T_j.\label{eq:factor-parameters}
\end{gather}
\end{lemma}

The proof of Lemma~\ref{lem:parameters-general} is in Section~\ref{sec:support-parameters-general}. Define
\begin{equation}\label{eq:failure-counts-general}
 M_1=1,\qquad M_j=1+L_jM_{j-1}=\sum_{i=1}^j\frac{T_j}{T_i}\quad(2\le j<k).
\end{equation}
Here $M_j$ counts a level-$j$ arrival block and all its constituent lower-level blocks.
A job's arrival block always refers to its original arrival to the switch. Each batch quantity below refers to one arbitrary arrival block; we add an index only when distinguishing several blocks.

\paragraph{The $k$ contributions.}
Let $X_j(t)$ count the level-$j$ jobs at the beginning of slot $t$. These levels partition the jobs present, so
\begin{equation}\label{eq:stage-decomposition-general}
 Q_\Sigma(t)=\sum_{j=1}^kX_j(t).
\end{equation}
We will show that each contribution is $O_k(nh\varepsilon^{-\alpha_k})$.

\subsection{Analysis of level-1 jobs}\label{sec:level-one-general}
Retain the definitions of $G_1$, $F_1$, $R_1$, $J_1$, and $D_1^{(v)}$ from Section~\ref{sec:level-one}, with $R_1$ now including jobs that started at any level above 1, in addition to jobs promoted to level 2. The bounds in Lemmas~\ref{lem:queue-length-one}, \ref{lem:success-one}, and~\ref{lem:population-one}, as well as the identity~\eqref{eq:leftovers-one}, hold unchanged.

The population and residual-count arguments are unchanged. For the edge distribution and failure bound, replace $\theta$ by $\tau_1$ and the mean degree by $\mu_1=\rho T_1-m_1$; Lemma~\ref{lem:parameters-general} supplies the same factor-test and concentration inequalities. The final estimate $T_1e^{-40h}\le a_1$ follows as before from~\eqref{eq:block-scales} and $1<\alpha_k\le4/3$. This also covers $k=2$, when $\tau_1=m_1=0$.

\subsection{Analysis of level-\texorpdfstring{$j$}{j} jobs, \texorpdfstring{$2\le j<k$}{2 <= j < k}}\label{sec:levels-general}
We proceed by induction on $j$, using the bound on $J_{j-1}$ established at the preceding level. First, we bound the number of level-$j$ jobs in the switch. There are no intermediate levels when $k=2$.

\begin{lemma}\label{lem:queue-length-general}
Starting from an empty switch, for every $2\le j<k$,
\[
 \sup_{t\ge0}\E{X_j(t)}\le4nL_ja_{j-1}.
\]
\end{lemma}
\begin{proof}
Fix a slot $t$. Every level-$j$ job is served or promoted by the last level-$j$ service slot of the next level-$j$ block after its original arrival. Such jobs therefore arrived during the current or preceding level-$j$ block, together containing at most $2L_j$ level-$(j-1)$ blocks. Each belongs to the graph $R_{j-1}$ for its original level-$(j-1)$ arrival block; jobs starting at level $j$ are included because no lower level serves them. Lemma~\ref{lem:population-one}(ii) for $j=2$, and the induction hypothesis for $j>2$, give $\E{J_{j-1}}\le2na_{j-1}$ and hence
\[
 \E{X_j(t)}\le(2L_j)\cdot(2na_{j-1})=4nL_ja_{j-1}.
\]
Counting all jobs in these graphs only increases the bound. At startup, there may be fewer blocks to count.
\end{proof}

Next, we will bound the number of jobs from a level-$j$ arrival block that enter higher levels. Consider the next level-$j$ block, and let $G_j$ be the graph of jobs eligible just before its first level-$j$ service slot. For each constituent level-$(j-1)$ arrival block $c$ of the previous block, let $R_{j-1,c}$ and $F_{j-1,c}$ be its graph of jobs not served at lower levels and its idle service-slot count. Then $G_j$ consists percisely of the jobs in these graphs whose initial levels are at most $j$.

Let $\mathcal F^{(j)}$ be the event that $G_j$ has no $S_j$-factor, and define
\[
 F_j=\sum_{c=1}^{L_j}F_{j-1,c}+S_j\ind_{\mathcal F^{(j)}}.
\]
Thus $F_j$ counts the level-$1,\ldots,j$ service slots for these jobs that stay idle because their factor tests fail. Later levels have fixed service counts, so these missed departures must be included when counting the jobs entering higher levels.

As at level 2 in the three-level proof, the jobs that start at level $j$ provide edges spread across the input--output pairs. Together with degree balance, these jobs ensure an $S_j$-factor with high probability. To account also for lower-level failures, let $\mathcal F_b^{(j)}$ denote the failure for arrival block $b$ and define
\[
 \mathcal G_b^{(1)}=\mathcal F_b^{(1)},\qquad
 \mathcal G_b^{(j)}=\mathcal F_b^{(j)}\cup
 \bigcup_{c\in\mathcal J_b}\mathcal G_c^{(j-1)},
\]
where $\mathcal J_b$ indexes the $L_j$ constituent level-$(j-1)$ blocks.

\begin{lemma}[Level-$j$ failures and idle service slots]\label{lem:success-general}
For every $2\le j<k$, uniformly over the blocks and for all sufficiently large $n$,
\[
 \Pp(\mathcal F_b^{(j)})\le\Pp(\mathcal G_b^{(j)})\le M_je^{-40h},
 \qquad \E{F_j}\le T_jM_je^{-40h}\le a_j.
\]
\end{lemma}

Lastly, we relate the idle service slots to the jobs entering higher levels. Let $R_j$ be the graph of jobs from the previous level-$j$ block that remain in the switch, measured immediately after the last level-$j$ service slot of the current block, and set
\[
 J_j=|E(R_j)|,\qquad D_j^{(v)}=\deg_{R_j}(v).
\]
Thus $R_j$ includes jobs promoted from level $j$ and jobs that started above level $j$. As for $R_1$, subsequent service does not remove jobs from this graph.

We can write $D_j^{(v)}$ explicitly as follows. Let $N^{(v)}(T_j)$ denote the total number of job arrivals at port $v$ during the previous level-$j$ block. Then
\begin{equation}\label{eq:residual-general}
 D_j^{(v)}=N^{(v)}(T_j)-\rho T_j+a_j+F_j.
\end{equation}
Indeed, by~\eqref{eq:planned-removals}, service at levels $1,\ldots,j$ removes $\sum_{i=1}^j(T_j/T_i)S_i-F_j=\rho T_j-a_j-F_j$ jobs at every port, and no higher-level service has yet been available to these jobs at the instant defining $R_j$.
The following lemma bounds their counts at each port and their total number.

\begin{lemma}\label{lem:population-general}
Starting from an empty switch, for every $2\le j<k$ and every level-$j$ arrival block,
\begin{enumerate}[itemsep=0.5em]
\item[$\mathrm{(i)}$] At every port $v$, $\E{D_j^{(v)}}\le2a_j$.
\item[$\mathrm{(ii)}$] $\E{J_j}\le2na_j$.
\end{enumerate}
\end{lemma}

The proofs of Lemmas~\ref{lem:success-general} and~\ref{lem:population-general} are in Section~\ref{sec:support-levels-general}.

\subsection{Analysis of level-\texorpdfstring{$k$}{k} jobs}\label{sec:terminal-general}
We now bound the number of level-$k$ jobs in the switch. Upon promotion to level $k$, a job waits to be batched with all other jobs from its original level-$k$ arrival block that reach level $k$. Just before the first level-$k$ service slot of the next level-$k$ block, these jobs are partitioned into matchings, which join the back of the queue of matchings.

To study this queue, index level-$k$ arrival blocks by $b=0,1,\ldots$, and let $A_k(b)$ be the number of matchings formed from the jobs of block $b$ that reach level $k$. Since the same rules are used in every block, these counts have the same distribution, and we use $A_k$ to denote a generic count.
Let $u_k$ be the offset of the first level-$k$ service slot within a level-$(k-1)$ block, where $0<u_k<T_{k-1}$. The batch from block $b$ enters the queue at $(b+1)T_k+u_k$. Let $Q_k(b)$ count the matchings already present just before this addition. Between this addition and the next, there are precisely $S_k=L_ka_{k-1}+\varepsilon T_k$ level-$k$ service slots and no other additions. Therefore, starting from an empty switch, the queue evolves as
\begin{equation}\label{eq:matching-recursion-general}
 Q_k(0)=0,\qquad Q_k(b+1)=\bigl(Q_k(b)+A_k(b)-S_k\bigr)^+.
\end{equation}
We will bound $q_k=\sup_{b\ge0}\E{Q_k(b)}$. As in the three-level proof, the following lemma relates the number of level-$k$ jobs to the number of matchings formed from each arrival block and the number already in the queue.

\begin{lemma}\label{lem:population-terminal-general}
Starting from an empty switch,
\begin{equation}\label{eq:level-k-counting}
 \sup_{t\geq 0} \, \E{X_k(t)}
 \le4nL_ka_{k-1}+n\bigl(q_k+\E{A_k}\bigr).
\end{equation}
\end{lemma}
\begin{proof}
Fix a slot $t$. We count the level-$k$ jobs waiting to be added to the queue of matchings, as well as the jobs already in the queue.

\par\addvspace{\medskipamount}
\noindent\emph{Jobs waiting to be added to the queue of matchings.}
These jobs arrived during the current or preceding level-$k$ block, together containing at most $2L_k$ level-$(k-1)$ blocks. Every such job belongs to its original level-$(k-1)$ block's $R_{k-1}$. Lemma~\ref{lem:population-one}(ii) for $k=2$, and Lemma~\ref{lem:population-general}(ii) for $k\ge3$, therefore bound their expected number by
\[
 (2L_k)\cdot(2na_{k-1})=4nL_ka_{k-1}.
\]
Counting all jobs in these graphs only increases the bound, and there are fewer blocks to count if the preceding block does not exist.

\par\addvspace{\medskipamount}
\noindent\emph{Jobs already in the queue of matchings.}
Let $\ell=\lfloor t/T_k\rfloor$. The queue is empty when $\ell=0$. For $\ell\ge1$, before time $\ell T_k+u_k$, there have been no additions or level-$k$ service slots in the current block, so the queue contains $Q_k(\ell-1)$ matchings. At that time, $A_k(\ell-1)$ matchings are added. Since there are no further additions during the block, the number of matchings can only decrease afterward. Each matching contains at most $n$ jobs, so the expected number of jobs in this group is at most
\[
 n\E{Q_k(\ell-1)+A_k(\ell-1)}\le n\bigl(q_k+\E{A_k}\bigr).
\]
The same bound holds for $\ell=0$, since the queue is empty. Adding the two bounds and taking the supremum over $t$ proves~\eqref{eq:level-k-counting}.
\end{proof}

It remains to bound $\E{A_k}$ and $q_k$. We use the following two lemmas, whose proofs are in Section~\ref{sec:support-terminal-general}.

\begin{lemma}[Matching arrivals and service slack]\label{lem:matching-general}
Let $\delta_k=S_k-\E{A_k}$. Then
\begin{enumerate}[itemsep=0.5em]
\item[$\mathrm{(i)}$] $0\le\E{A_k}-L_ka_{k-1}\le\varepsilon T_k/2$.
\item[$\mathrm{(ii)}$] $\Var{A_k}\le34T_kh$.
\item[$\mathrm{(iii)}$] $\varepsilon T_k/2\le\delta_k\le\varepsilon T_k$.
\end{enumerate}
\end{lemma}

Part (i) directly bounds $\E{A_k}$ in~\eqref{eq:level-k-counting}. To bound $q_k$, we use the scalar queue bound $q_k\le\Var{A_k}/(2\delta_k)+\delta_k/2$, which requires the variance estimate in (ii) and the bound on service slack in (iii). This gives the following lemma.

\begin{lemma}[Queue of matchings]\label{lem:queue-general}
The variables $A_k(b)$ are iid and satisfy $0\le A_k(b)\le nT_k$. Starting from an empty switch,
\begin{equation}\label{eq:matching-queue-bound-general}
 q_k\le\frac{34h}{\varepsilon}+\frac{\varepsilon T_k}{2}
 \le\left(34+\frac{(2K)^k}{2}\right)\frac h\varepsilon.
\end{equation}
\end{lemma}

Consequently, substituting the bounds on $\E{A_k}$ and $q_k$ into Lemma~\ref{lem:population-terminal-general} gives
\begin{equation}\label{eq:level-k-population-bound}
 \sup_{t\ge0}\E{X_k(t)}
 \le n\bigl(5L_ka_{k-1}+34h/\varepsilon+\varepsilon T_k\bigr).
\end{equation}

\subsection{Completing the proof for general \texorpdfstring{$k$}{k}}
Combining Lemmas~\ref{lem:queue-length-one} and~\ref{lem:queue-length-general} with~\eqref{eq:level-k-population-bound} in the decomposition~\eqref{eq:stage-decomposition-general} gives
\[
 \E{Q_\Sigma(t)}
 \le n\left(2T_1+4\sum_{j=2}^{k-1}L_ja_{j-1}
       +5L_ka_{k-1}+\frac{34h}{\varepsilon}+\varepsilon T_k\right).
\]
The intermediate-level sum is empty when $k=2$. Finally, $\beta_k=2$ and
\[
 \frac{\alpha_k}{2^{j-1}}+\frac{\beta_{j-1}}2=\alpha_k
 \qquad(2\le j\le k).
\]
Together with~\eqref{eq:block-scales}, these identities give
\[
 \begin{aligned}
 T_1&=O_k(h\varepsilon^{-\alpha_k}),\qquad
 \varepsilon T_k=O_k(h/\varepsilon),\\
 L_ja_{j-1}
 &=O_k\!\left(\varepsilon^{-\alpha_k/2^{j-1}}\sqrt{T_{j-1}h}\right)
 =O_k(h\varepsilon^{-\alpha_k})\quad(2\le j\le k).
 \end{aligned}
\]
Since $k$ is fixed, $h/\varepsilon\le h\varepsilon^{-\alpha_k}$, and~\eqref{eq:log-factor} holds,
\[
 \sup_{t\ge0}\E{Q_\Sigma(t)}
 =O_k\!\left(n\log(n/\varepsilon)\varepsilon^{-\alpha_k}\right).
\]
This proves Theorem~\ref{thm:main} for every fixed $k\ge2$.

\subsection{Supporting proofs for general~\texorpdfstring{$k$}{}} \label{sec:support-general}
We present the supporting proofs in the same order as the levels, using the factor and matching-arrival estimates in Section~\ref{sec:supporting-proofs}. First, we verify the parameter bounds used in the factor tests.

\subsubsection{Parameters for the factor tests}\label{sec:support-parameters-general}
\begin{proof}[Proof of Lemma~\ref{lem:parameters-general}]
\noindent\emph{Jobs starting above level $j$.}
Since $C\sqrt{T_jh}\le a_j<(C+1)\sqrt{T_jh}$,
\[
 \frac{a_{j+1}/T_{j+1}}{a_j/T_j}
 \le\frac{C+1}{C\sqrt{L_{j+1}}}
 \le\frac{109}{1296}<\frac1{10}\qquad(1\le j\le k-2).
\]
Substituting into~\eqref{eq:future-reserves} gives
\[
 m_j=\frac{T_j}{2}\sum_{r=j}^{k-2}\frac{a_r}{T_r}
 \le\frac{a_j}{2}\sum_{r=0}^{\infty}10^{-r}
 =\frac{5a_j}{9}\le\frac{2a_j}{3}.
\]

\par\addvspace{\medskipamount}
\noindent\emph{Validity of the parameters.}
By~\eqref{eq:reserve-bound} and~\eqref{eq:block-scales},
\[
 \sum_{r=1}^{k-2}\theta_r=\tau_1
 \le\frac{5a_1}{9\rho T_1}\to0,
 \qquad \frac{d_1}{a_1}=\frac{\rho T_1}{a_1}\to\infty.
\]
For $2\le j<k$,
\[
 \frac{d_j}{a_j}
 =\frac{L_ja_{j-1}}{a_j}
 \ge\frac{C\sqrt{L_j}}{C+1}
 \ge\frac{1296}{109}>8.
\]
Thus the initial-level distribution is valid and $S_j=d_j-a_j>0$ for every $j<k$ once $n$ is sufficiently large.
By~\eqref{eq:remaining-slots}, each allocation leaves $\varepsilon T_j+a_j>0$ slots, so the service counts are feasible.

\par\addvspace{\medskipamount}
\noindent\emph{Parameters for the degree and cut estimates.}
Using~\eqref{eq:reserve-bound} and $d_j/a_j>8$,
\[
\begin{aligned}
 \mu_j=d_j-m_j
 &\ge d_j-\frac{2a_j}{3}>\frac{22a_j}{3}\ge24(a_j/9),\\
 \mu_j-3(a_j/9)
 &=d_j-m_j-\frac{a_j}{3}\ge d_j-a_j=S_j.
\end{aligned}
\]
Moreover, $\mu_j\ge d_j/2$, and~\eqref{eq:block-scales} gives
\[
 \frac{\mu_j}{h}\ge\frac{d_j}{2h}\to\infty,
 \qquad
 \frac{\nu_j}{a_j}\ge\frac{\rho T_j}{a_j}-\frac23\to\infty.
\]
Together with $\nu_j\le\rho T_j\le T_j$, these limits prove the remaining inequalities in~\eqref{eq:factor-parameters}. Since $k$ is fixed, they hold simultaneously for all $j<k$.
\end{proof}

\subsubsection{Proofs for level-\texorpdfstring{$j$}{j} jobs, \texorpdfstring{$2\le j<k$}{2 <= j < k}}\label{sec:support-levels-general}
\begin{proof}[Proof of Lemma~\ref{lem:success-general}]
Fix $2\le j<k$ and a level-$j$ arrival block $b$. Let $N_{v,j}$ count its original arrivals at port $v$ whose initial level is at most $j$, and let $R$ be the graph of its current level-$j$ jobs at their service test. Then
\[
 N_{v,j}\sim\Binomial(nT_j,\rho(1-\tau_j)/n),\qquad \E{N_{v,j}}=\nu_j.
\]
If all lower-level services succeed, their removals concern jobs with $M\le j$, so
\[
    \deg_R(v)=N_{v,j}-(\rho T_j-d_j)=\mu_j+(N_{v,j}-\nu_j).
\]
Let $H_j$ consist of the arrivals with $M=j$. These jobs start at level $j$ and wait for this batch's test, so $H_j\subseteq R$ on every outcome. In particular, lower-level failures do not change their level. Its edge multiplicities are independent $\Binomial(T_j,\rho\theta_{j-1}/n)$ variables, with
\[
    \mathbb{E} \big[\deg_{H_j}(v) \big]=\rho T_j\theta_{j-1}=d_j/2\ge\mu_j/2.
\]
Let $\mathcal D_j$ be the event that $|N_{v,j}-\nu_j|\le a_j/9$ at every port, and let $\mathcal C_j$ be the cut event~\eqref{eq:protectedcut} for $H_j$ with $d=\mu_j$.
As in Lemmas~\ref{lem:success-one} and~\ref{lem:success-two}, the residual-degree identity and the factor conclusion of Lemma~\ref{lem:randomprotected} show that $\mathcal D_j\cap\mathcal C_j$, together with success at all lower levels, implies an $S_j$-factor.
The standard Chernoff bound (Appendix~\ref{sec:chernoff}), combined with a union bound over the $2n$ ports and Lemma~\ref{lem:randomprotected}, yields
\[
     \Pp(\mathcal D_j^c)+\Pp(\mathcal C_j^c)
     \leq 4n\exp\!\left(-\frac{a_j^2}{243\nu_j}\right)+4e^{-\mu_j/96}
     \leq 4ne^{-48h}+4e^{-48h}\le e^{-40h}.
\]
Lemma~\ref{lem:parameters-general} verifies the hypotheses of~\eqref{eq:chernoff} and Lemma~\ref{lem:randomprotected}. It also gives $\nu_j \leq T_j$ and $\mu_j/96\ge48h$; the last line uses $a_j^2 \geq C^2T_jh$, $C^2/243=48$, $n\le e^h$, and $h\geq 1$.
As in the proof of Lemma~\ref{lem:success-two},
\[
     \mathcal G_b^{(j)}\subseteq
     \left(\bigcup_{c\in\mathcal J_b}\mathcal G_c^{(j-1)}\right)
     \cup\mathcal D_j^c\cup\mathcal C_j^c.
\]
Lemma~\ref{lem:success-one} gives $\Pp(\mathcal G_b^{(1)})\le e^{-40h}$. Induction and the union bound therefore yield
\[
     \Pp(\mathcal G_b^{(j)})
     \leq L_jM_{j-1}e^{-40h}+e^{-40h}=M_je^{-40h}.
\]
Since $\mathcal F_b^{(j)}\subseteq\mathcal G_b^{(j)}$, this proves the probability bounds. No independence between the degree, cut, and earlier failure events is needed.

For a level-$j$ block $b$, every failure indicator in $F_j$ is at most $\ind_{\mathcal G_b^{(j)}}$. Unrolling the definition of $F_j$ and using~\eqref{eq:planned-removals} gives
\[
 0\le F_j\le(\rho T_j-a_j)\ind_{\mathcal G_b^{(j)}}
 \le T_j\ind_{\mathcal G_b^{(j)}},
 \qquad \E{F_j}\le T_jM_je^{-40h}.
\]
Since $M_j\le jT_j/T_1$ and $a_j\ge C\sqrt{T_jh}$,~\eqref{eq:block-scales} gives
\[
 \frac{T_jM_je^{-40h}}{a_j}
 \le\frac{jT_j^{3/2}}{CT_1\sqrt h}\,e^{-40h}
 =O_k\!\left(\varepsilon^{-(3\beta_j/2-\alpha_k)}e^{-40h}\right)
 =O_k(\varepsilon^{-2}e^{-4h})
 =O_k(\varepsilon^2/n^4)\to0.
\]
Here $\beta_j\le2$, $\alpha_k\ge1$, and $e^{-4h}=(\varepsilon/(Kn))^4$.
Thus $\E{F_j}\le a_j$ for all sufficiently large $n$.
\end{proof}

\begin{proof}[Proof of Lemma~\ref{lem:population-general}]
Taking expectations in~\eqref{eq:residual-general}, with $\E{N^{(v)}(T_j)}=\rho T_j$, and applying Lemma~\ref{lem:success-general} gives
\[
 \E{D_j^{(v)}}=a_j+\E{F_j}\le2a_j.
\]
This proves (i). Summing over the input ports gives (ii).
\end{proof}

\subsubsection{Proofs for level-\texorpdfstring{$k$}{k} jobs}\label{sec:support-terminal-general}
\begin{proof}[Proof of Lemma~\ref{lem:matching-general}]
For one level-$k$ arrival block, let $D_k^{(v)}$ count its jobs at port $v$ that reach level $k$, and set
\[
 F_k=\sum_{c=1}^{L_k}F_{k-1,c},\qquad
 Z_k=\max_{v\in\mathcal I\cup\mathcal O}N^{(v)}(T_k)-\rho T_k.
\]
Summing the port-count identities of the $L_k$ constituent level-$(k-1)$ blocks gives the first identity below. By K\H{o}nig's edge-coloring theorem, $A_k=\max_vD_k^{(v)}$, which gives the second:
\begin{equation}\label{eq:terminal-residual-general}
 \begin{aligned}
 D_k^{(v)}&=N^{(v)}(T_k)-\rho T_k+L_ka_{k-1}+F_k,\\
 A_k&=L_ka_{k-1}+Z_k+F_k.
 \end{aligned}
\end{equation}
Since these batches reserve $\rho T_k-L_ka_{k-1}\le T_k$ service slots in total, $0\le F_k\le T_k$. Lemma~\ref{lem:success-one} for $k=2$, and Lemma~\ref{lem:success-general} for $k\ge3$, give
\[
 \E{F_k}\le L_kT_{k-1}M_{k-1}e^{-40h}
 =T_kM_{k-1}e^{-40h}.
\]
By~\eqref{eq:failure-counts-general} and~\eqref{eq:block-scales},
\[
 M_{k-1}\le(k-1)\frac{T_{k-1}}{T_1}\le(k-1)\frac{T_k}{T_1}
 =O_k\!\left(\varepsilon^{-(2-\alpha_k)}\right)
 =O_k(\varepsilon^{-1}).
\]
Since $T_k=O_k(h\varepsilon^{-2})$, it follows that
\begin{equation}\label{eq:failure-moments-general}
 \E{F_k}\le T_kM_{k-1}e^{-40h}
 =O_k(h\varepsilon^{-3}e^{-6h})
 =O_k(h\varepsilon^3/n^6)\le h.
\end{equation}
We used $e^{-40h}\le e^{-6h}=(\varepsilon/(Kn))^6$; the last inequality holds for sufficiently large $n$ because $k$ is fixed.

Lemma~\ref{lem:matching-moments} applies to~\eqref{eq:terminal-residual-general} with $T=T_k$, $D=(2K)^k$, $c=L_ka_{k-1}$, and $F=F_k$. Indeed,~\eqref{eq:block-scales} gives~\eqref{eq:terminal-length}, and the preceding bounds give $0\le F_k\le T_k$ and $\E{F_k}\le h$. The constant $D$ is fixed because $k$ is fixed. No independence between $Z_k$ and $F_k$ is needed.

\par\addvspace{\medskipamount}
\noindent (i)--(ii)
The corresponding parts of Lemma~\ref{lem:matching-moments} give the mean and variance bounds.

\par\addvspace{\medskipamount}
\noindent (iii)
Since $S_k=L_ka_{k-1}+\varepsilon T_k$, part (i) gives
\[
 \delta_k=S_k-\E{A_k}
 =\varepsilon T_k-\bigl(\E{A_k}-L_ka_{k-1}\bigr)
 \in[\varepsilon T_k/2,\varepsilon T_k].\qedhere
\]
\end{proof}

\begin{proof}[Proof of Lemma~\ref{lem:queue-general}]
The slot allocation is the same in every block, with $0<u_k<T_{k-1}$. The insertion times are $(b+1)T_k+u_k$, and exactly $S_k$ level-$k$ service slots occur between consecutive insertions, with no other additions. This also covers $k=2$, when the only preceding service is level 1. Since the switch starts empty, $Q_k(0)=0$. Serving one matching in each level-$k$ service slot whenever the queue is nonempty gives~\eqref{eq:matching-recursion-general}.
Each $A_k(b)$ depends only on its original block's arrivals and initial levels through the policy's fixed rules. Thus the matching counts are iid. Also,
\[
 0\le A_k(b)\le\max_vN^{(v)}(T_k)\le nT_k.
\]
Set $\delta_k=S_k-\E{A_k}>0$, where positivity follows from Lemma~\ref{lem:matching-general}(iii). Applying Lemma~\ref{lem:scalar-queue} with $\lambda=\E{A_k}$, $v=\Var{A_k}$, and $B=S_k$, and then the estimates in Lemma~\ref{lem:matching-general}(ii)--(iii), gives
\[
 q_k\le\frac{\Var{A_k}}{2\delta_k}+\frac{\delta_k}{2}
 \le\frac{34h}{\varepsilon}+\frac{\varepsilon T_k}{2}
 \le\left(34+\frac{(2K)^k}{2}\right)\frac h\varepsilon.
\]
The last inequality uses $T_k\le(2K)^kh\varepsilon^{-2}$.
\end{proof}

\section{Table of notation} \label{apx:table}
\begin{table}[!t]
    \centering
    \caption{Notation for the three-level analysis.}
    \label{tab:three-level-notation}
    \small
    \renewcommand{\arraystretch}{1.08}
    \begin{tabular}{@{}p{0.20\textwidth}p{0.77\textwidth}@{}}
        \toprule
        \textbf{Notation} & \textbf{Meaning} \\
        \midrule
        \multicolumn{2}{@{}l}{\textbf{Common notation}} \\
        \addlinespace[2pt]
        $T_j$
        & Number of slots in a level-$j$ block, i.e., $L_1\cdots L_j$. \\
        $L_j$
        & Number of level-$(j-1)$ blocks in a level-$j$ block for $j=2,3$; $L_1$ is the number of slots in a level-1 block. \\
        $S_j$
        & Number of level-$j$ service slots in a level-$j$ block. \\
        $h$
        & A logarithmic factor $\log(Kn/\varepsilon)$. \\
        $\theta$
        & Probability that an arriving job starts at level 2, i.e., $a_1/(2\rho T_1)$. \\
        $X_j(t)$
        & Number of unserved level-$j$ jobs at the beginning of slot $t$. \\
        $N^{(v)}(T)$ 
        & Number of arrivals at port $v$ during a block of $T$ slots. \\
        \midrule
        \multicolumn{2}{@{}l}{\textbf{Level 1} (Section~\ref{sec:level-one})} \\
        \addlinespace[2pt]
        $a_1$
        & Amount by which $S_1$ is below $\rho T_1$, i.e., $\rho T_1-S_1$. \\
        $G_1$
        & Graph of jobs that arrive in a level-1 block with initial level 1. \\
        $\mathcal F^{(1)}$
        & Event that $G_1$ has no $S_1$-factor. \\
        $F_1$
        & Number of level-1 service slots left idle because the factor test for $G_1$ fails. \\
        $R_1$
        & Graph of jobs from a level-1 arrival block that start at level 2 or are promoted to level 2. \\
        $J_1$
        & Number of jobs in the graph $R_1$, i.e., $|E(R_1)|$. \\
        $D_1^{(v)}$
        & Number of jobs at port $v$ in the graph $R_1$, i.e., $\deg_{R_1}(v)$. \\
        \midrule
        \multicolumn{2}{@{}l}{\textbf{Level 2} (Section~\ref{sec:level-two})} \\
        \addlinespace[2pt]
        $a_2$
        & Amount by which $S_2$ is below $L_2a_1$, i.e., $L_2a_1-S_2$. \\
        $G_2$
        & Graph of level-2 jobs from a level-2 arrival block, collected for service in the next level-2 block. \\
        $\mathcal F^{(2)}$
        & Event that $G_2$ has no $S_2$-factor. \\
        $F_2$
        & Number of level-1 and level-2 service slots left idle by failed factor tests for jobs from a level-2 arrival block. \\
        $R_2$
        & Graph of jobs from a level-2 arrival block that are promoted to level 3. \\
        $J_2$
        & Number of jobs in the graph $R_2$, i.e., $|E(R_2)|$. \\
        $D_2^{(v)}$
        & Number of jobs at port $v$ in the graph $R_2$, i.e., $\deg_{R_2}(v)$. \\
        \midrule
        \multicolumn{2}{@{}l}{\textbf{Level 3} (Sections~\ref{sec:matching-queue} and~\ref{sec:support-three})} \\
        \addlinespace[2pt]
        $A_3(b)$
        & Number of matchings formed from jobs of level-3 arrival block $b$ that reach level 3. \\
        $Q_3(b)$
        & Number of matchings in the queue just before $A_3(b)$ matchings are added. \\
        $q_3$
        & $\sup_{b\ge0}\E{Q_3(b)}$. \\
        $\delta_3$
        & $S_3-\E{A_3(b)}$. \\
        $u$
        & Time offset of the first level-3 service slot within a level-2 block. \\
        $D_3^{(v)}$
        & Number of jobs at port $v$ from a level-3 arrival block that reach level 3. \\
        $Z_3$
        & Maximum number of arrivals at a port during a level-3 block, minus $\rho T_3$. \\
        $F_3$
        & Number of level-1 and level-2 service slots left idle by failed factor tests for jobs from a level-3 arrival block. \\
        \bottomrule
    \end{tabular}
\end{table}
\end{document}